\documentclass[1p,number,preprint]{elsarticle}
\usepackage{amsfonts}
\usepackage{amsmath}
\usepackage{amssymb}
\usepackage{amsthm}
\usepackage[english]{babel}
\usepackage{complexity}
\usepackage{enumerate}
\usepackage{mathrsfs}
\usepackage{graphics}
\usepackage[pagewise]{lineno}
\usepackage{mathtools}
\usepackage{microtype}
\usepackage[defaultlines=4,all]{nowidow}
\usepackage[many]{tcolorbox}
\usepackage{xcolor} 
\usepackage{xspace}
\usepackage{soul} 

\usetikzlibrary{calc,arrows, automata, fit, shapes, arrows.meta, positioning}

\tcolorboxenvironment{theorem}{
    colframe=blue, sharp corners=west, 
    colback=blue!60!cyan!20!white,
	rightrule=0pt,
	toprule=0pt,
	bottomrule=0pt,
	top=0pt,
	right=2pt,
	bottom=0pt,
 }

 \tcolorboxenvironment{proposition}{
    colframe=blue, sharp corners=west, 
    colback=blue!50!cyan!20!white,
	rightrule=0pt,
	toprule=0pt,
	bottomrule=0pt,
	top=0pt,
	right=2pt,
	bottom=0pt,
 }

\tcolorboxenvironment{corollary}{
    colframe=cyan, sharp corners=west, 
    colback=blue!30!cyan!10!white,
	rightrule=0pt,
	toprule=0pt,
	bottomrule=0pt,
	top=0pt,
	right=2pt,
	bottom=0pt,
 }

 \tcolorboxenvironment{observation}{
    colframe=cyan, sharp corners=west, 
    colback=blue!30!cyan!10!white,
	rightrule=0pt,
	toprule=0pt,
	bottomrule=0pt,
	top=0pt,
	right=2pt,
	bottom=0pt,
 }

\tcolorboxenvironment{lemma}{
    colframe=cyan, sharp corners=west,
    colback=blue!40!cyan!20!white,
	rightrule=0pt,
	toprule=0pt,
	bottomrule=0pt,
	top=0pt,
	right=2pt,
	bottom=0pt,
 }

 \tcolorboxenvironment{result}{
    colframe=yellow, sharp corners=west, 
    colback=yellow!10,
	rightrule=0pt,
	toprule=0pt,
	bottomrule=0pt,
	top=0pt,
	right=2pt,
	bottom=0pt,
 }

\usepackage{hyperref}

\usepackage{enumitem}
\usepackage{mathrsfs}

\usepackage{amsmath}
\usepackage{amssymb}
\usepackage{amsthm}
\usepackage{thm-restate}
\usepackage{mathtools}

\usepackage{float}
\usepackage{wrapfig}

\usepackage{tikz}
\usepackage{cancel}
\usepackage{mdframed}
\usepackage{xspace}
\usepackage[textsize=normalsize]{todonotes}

\usepackage{caption}
\usepackage{subcaption}

\newcommand{\Oof}{\mathcal{O}}
\newcommand{\N}{\mathbb{N}}
\renewcommand{\phi}{\varphi}

\theoremstyle{definition}
\newtheorem{claim}{Claim}
\newtheorem{claimproof}{Proof of Claim}

\usepackage[noabbrev,capitalise,nameinlink]{cleveref}
\hypersetup{
    colorlinks,
    linkcolor={blue!80!black},
    citecolor={blue!80!black},
    urlcolor={blue!80!black}
}

\crefname{observation}{Observation}{Observation}
\crefalias{enumi}{observation}
\crefname{claim}{Claim}{Claim}

\newcommand{\Oh}{\mathcal{O}}

\newcommand{\Cc}{\mathscr{C}}

\newcommand{\Tt}{\mathcal{T}}

\newcommand{\bag}{\mathsf{bag}}
\newcommand{\cone}{\mathsf{cone}}
\newcommand{\cmp}{\mathsf{comp}}
\newcommand{\mrg}{\mathsf{mrg}}
\newcommand{\adh}{\mathsf{adh}}

\newcommand{\parent}{\mathsf{parent}}
\newcommand{\children}{\mathsf{children}}

\renewcommand{\FO}{\ensuremath{\mathsf{FO}}}

\newcommand{\MSO}{\ensuremath{\mathsf{MSO}}}
\newcommand{\CMSO}{\ensuremath{\mathsf{CMSO}}}
\renewcommand{\FPT}{\ensuremath{\mathsf{FPT}}}

\newcommand{\FOMSO}{{\FO(\MSO(\preceq,A)\cup\Sigma)}}
\renewcommand{\DP}{\mathsf{dp}}
\newcommand{\FODP}{\ensuremath{\FO\hspace{0.8pt}\raisebox{-0.3pt}{\texttt{+}}\hspace{0.8pt}\DP}\xspace}

\renewcommand{\leq}{\leqslant}
\renewcommand{\geq}{\geqslant}
\renewcommand{\le}{\leqslant}
\renewcommand{\ge}{\geqslant}

\newtheorem{lemma}{Lemma}[section]
\newtheorem{corollary}[lemma]{Corollary}
\newtheorem{theorem}[lemma]{Theorem}

\newtheorem{proposition}[lemma]{Proposition}

\newtheorem{definition}[lemma]{Definition}
\newtheorem{observation}[lemma]{Observation}

\usepackage{color}

\definecolor{MidnightBlack}{rgb}{0.1,0.1,.34}
\definecolor{MidnightBlue}{rgb}{0.1,0.1,0.44}
\definecolor{Black}{rgb}{0,0, 0}
\definecolor{Blue}{rgb}{0, 0 ,1}
\definecolor{Red}{rgb}{1, 0 ,0}
\definecolor{White}{rgb}{1, 1, 1}
\definecolor{Grey}{rgb}{.6, .6, .6}
\definecolor{Mygreen}{rgb}{.0, .5, .0}
\definecolor{Yellow}{rgb}{.55,.55,0}
\definecolor{Mustard}{rgb}{1.0, 0.86, 0.35}
\definecolor{applegreen}{rgb}{0.55, 0.71, 0.0}
\definecolor{darkturquoise}{rgb}{0.0, 0.81, 0.82}
\definecolor{celestialblue}{rgb}{0.29, 0.59, 0.82}
\definecolor{green_yellow}{rgb}{0.68, 1.0, 0.18}
\definecolor{crimsonglory}{rgb}{0.75, 0.0, 0.2}
\definecolor{darkmagenta}{rgb}{0.30, 0.0, 0.30}
\definecolor{internationalorange}{rgb}{1.0, 0.31, 0.0}
\definecolor{internationalorange}{rgb}{1.0, 165/256, 0.0}
\definecolor{darkorange}{rgb}{1.0, 0.55, 0.0}
\definecolor{ao}{rgb}{0.0, 0.5, 0.0}
\definecolor{awesome}{rgb}{1.0, 0.13, 0.32}

\journal{arXiv}

\begin{document}

\begin{frontmatter}

\title{Model Checking Disjoint-Paths Logic \\ on Topological-Minor-Free Graph Classes\tnoteref{t1}}
\tnotetext[t1]{The results of this article appeared in the Proceedings of the 39th Annual ACM/IEEE Symposium on Logic in Computer Science (LICS 2024)~\cite{SchirrmacherSSTV24}.
This research was supported by the
French-German Collaboration ANR/DFG Project UTMA (ANR-20-CE92-0027, DFG 446200270) {\em ``Unifying Theories for Multivariate Algorithms''}.
The 4th author was also supported by the ANR project GODASse ANR-24-CE48-4377,  and by the Franco-Norwegian project PHC AURORA 2024 (Projet n°\! 51260WL).}

\author{Nicole Schirrmacher\texorpdfstring{$^{\star\star}$}{**}, Sebastian Siebertz\texorpdfstring{$^{\star\star}$}{**}, Giannos Stamoulis\texorpdfstring{$^{\star\star\star}$}{***}\\ Dimitrios M. Thilikos\texorpdfstring{$^{\star\star\star\star}$}{****}, Alexandre Vigny\texorpdfstring{$^{\star\star\star\star\star}$}{*****}}

\tnotetext[t2]{Nicole Schirrmacher, Sebastian Siebertz: University of Bremen, Germany, \texttt{$\{$schirrmacher, siebertz$\}$@uni-bremen.de}}
\tnotetext[t3]{\hspace{-4mm}$^{\star\star\star}$Giannos Stamoulis: Université Paris Cité, CNRS, IRIF, F-75013, Paris, France, \texttt{stamoulis@irif.fr}}
\tnotetext[t4]{\hspace{-4mm}$^{\star\star\star\star}$ Dimitrios M. Thilikos: LIRMM, Univ Montpellier, CNRS, France, \texttt{sedthilk@thilikos.info}}
\tnotetext[t5]{\hspace{-4mm}$^{\star\star\star\star\star}$ Alexandre Vigny: Université Clermont Auvergne, France, \texttt{alexandre.vigny@uca.fr}}

\begin{abstract}
    \noindent
    \emph{Disjoint-paths logic}, denoted $\FODP$, extends first-order logic (\FO) with atomic predicates $\DP_r[(x_1, y_1), \ldots ,(x_r, y_r)]$, expressing the existence of vertex-disjoint paths between $x_i$ and~$y_i$, for  $1\leq i\leq r$. We prove that for every graph class excluding some fixed graph as a topological minor, the model checking problem for $\FODP$ is fixed-parameter tractable. This essentially settles the question of tractable model checking for this logic on subgraph-closed classes, since the problem is hard on subgraph-closed classes not excluding a topological minor (assuming a further mild condition of efficiency of encoding).  
\end{abstract}

\end{frontmatter}


\section{Introduction}

The model checking problem for a logic $\mathcal{L}$ gets as input a structure and an $\mathcal{L}$-sentence and the question is to decide whether the sentence is true in the structure. Therefore, the model checking problem for $\mathcal{L}$ expresses all $\mathcal{L}$-definable problems. For this reason, tractability results for model checking problems are often called \emph{algorithmic meta theorems}, as they explain and unify tractability for all problems definable in the considered logic~$\mathcal{L}$. A prime example of an algorithmic meta theorem is Courcelle's theorem~\cite{Courcelle90} stating that every problem definable in monadic second-order logic ($\MSO$) can be solved in linear time on every graph class with bounded treewidth. An algorithmic meta theorem not only provides a quick way to establish tractability of problems but in many cases, its proof distills the essence of the algorithmic techniques required to solve them. Courcelle's theorem captures the decomposability of $\MSO$-definable problems and a corresponding dynamic programming approach over tree decompositions of small width. Courcelle's theorem was extended to graph classes with bounded cliquewidth~\cite{courcelle2000linear} and it is known that these are essentially the most general graph classes on which we can expect efficient $\MSO$ model \mbox{checking~\cite{ganian2014lower,kreutzer2010lower}}.

Also the first-order ($\FO$) model checking problem has received considerable attention in the literature, see e.g.~\cite{bonnet2022twin,bonnet2021twin,dawar2007locally,dreier2023first,dreier2024first,dvovrak2010deciding,eickmeyer2017fo,flum2001fixed,frick2001deciding,gajarsky2020first,ganian2013fo,GroheKS17,HlinenyPR17,seese1996linear,SiebertzV24}. Seese~\cite{seese1996linear} was the first to study the $\FO$ model checking problem on classes of graphs with bounded maximum degree.
The essence of his approach was to exploit the locality properties of~$\FO$, and, in some form, locality-based methods constitute the basis of all of the mentioned model checking results for $\FO$. Grohe,  Kreutzer, and Siebertz~\cite{GroheKS17} showed that the $\FO$ model checking problem is fixed-parameter tractable on nowhere dense graph classes and when considering subgraph-closed classes this result is optimal~\cite{dvovrak2010deciding}. 
In a recent breakthrough, it was shown that the problem is fixed-parameter tractable on classes with bounded twinwidth~\cite{bonnet2021twin} and, moreover, this result is optimal when considering classes of ordered structures~\cite{bonnet2022twin}.

While $\FO$ can express many interesting algorithmic properties, it also has some shortcomings. In particular, it cannot count and it can express \textsl{only} local properties. The first shortcoming led to the study of counting extensions, see e.g.~\cite{berkholz2018answering,dreier2021approximate,grohe2018first,kuske2017first,kuske2018gaifman,NesetrilOdMS24,torunczyk2020aggregate}. These meta theorems, in essence,  build again on locality properties that are shared by $\FO$ with counting extensions. The second shortcoming classically led to the study of transitive-closure logics and fixed-point logics, see e.g.~\cite{ebbinghaus1999finite,gradel2007finite,grohe2008logic,libkin2004elements}. However, even the model checking problem for the very restricted monadic transitive-closure logic~$\mathrm{TC}^1$ is most probably not fixed-parameter tractable even on planar graphs of maximum degree three~\cite{grohe2008logic}. Furthermore, these logics still fall short of being able to express many interesting algorithmic problems, involving  ``non-local'' queries, studied in contemporary algorithmics. A classic problem of this kind is the \textsc{Disjoint Paths} problem: \textsl{Given a graph $G$ and a set $\{(s_{1},t_{1}),\ldots,(s_{r},t_{r})\}$ of pairs of terminals, the question is whether $G$ contains vertex-disjoint paths joining $s_{i}$ and~$t_{i}$ for $1\leq i\leq r$.} 
Clearly, asking for a path joining pairs of terminals is not a local query as the size of such a path is unbounded. 

This has recently led to the study of new logics whose expressive power lies between~$\FO$ and~$\MSO$. \emph{Compound logic} combines $\FO$ and $\MSO$ and is designed to express a wide range of graph modification problems~\cite{FominGSST24}.
Its model checking problem is fixed-parameter tractable on classes of graphs with excluded minors and as its core, it combines the locality method for $\FO$ with the important algorithmic irrelevant vertex technique to eventually reduce the problem to a bounded treewidth graph. 
Another recently introduced formalism is  $\textsf{CMSO/tw}$  introduced in~\cite{SauST25}, where  second-order quantification 
is restricted to sets of bounded \textsl{bidimensionality}.\footnote{The \emph{bidimensionality} of a vertex set $X$ of a graph $G$ 
is the maximum $k$ for which $G$ can be contracted to an $(k\times k)$-grid in a way that for each vertex of the resulting grid at least one vertex of $X$ has been identified.}
This allows one to express properties that are not first-order definable while still enabling fixed-parameter tractable model checking on minor-closed graph classes.
Algorithmically, the approach combines logical locality with minor-based reductions that confine the relevant part of the instance to regions of bounded treewidth and it permits model checking in quadratic time on classes excluding some graph as a minor.
In follow-up work, we generalize the approach and study model checking for logics based on general \emph{annotated graph parameters} on topological-minor-free graph classes~\cite{Sau26}.
Yet another recently introduced logic is \emph{separator logic}, which extends $\FO$ by connectivity after vertex deletions~\cite{Bojanczyk21,SchirrmacherSV22} and which can express other interesting algorithmic problems such as elimination distance to \FO-definable graph classes.
It was proven in \cite{PilipczukSSTV22} that, for this logic, the model checking problem is fixed-parameter tractable on classes excluding a topological minor, and for subgraph-closed classes, this result cannot be extended to more general classes (assuming a further condition on the efficiency of encoding\footnote{We say that a class $\Cc$ \emph{admits efficient encoding of topological minors} if for every graph $H$ there exists $G \in \Cc$ such that $H$ is a topological minor of $G$, and, given $H$, such $G$ together with a suitable topological minor model can be computed in time polynomial in $|H|$.}).
This meta theorem essentially combines classical $\FO$ model checking with dynamic programming over decompositions into unbreakable parts. The required decompositions are provided by a result of  Cygan, Lokshtanov, Pilipczuk, Pilipczuk, and Saurabh~\cite{cygan2019minimum}. A key observation is that over highly connected graphs, connectivity can be reduced to queries of bounded length paths, and therefore becomes in fact \FO~expressible. 

In this work, we study \emph{disjoint-paths logic}, which was also introduced in~\cite{SchirrmacherSV22} as an extension of separator logic. Disjoint-paths logic, denoted $\FODP$, extends first-order logic~(\FO) with atomic predicates $\DP_r[(x_1, y_1), \ldots ,(x_r, y_r)]$ expressing the existence of vertex-disjoint paths\footnote{In~\cite{SchirrmacherSV22}, the predicate expresses the existence of \emph{internally vertex-disjoint paths}, however, it is easy to see that the two extensions lead to the same expressive power.} between $x_i$ and $y_i$, for  $1\leq i\leq r$. It can express many interesting algorithmic problems, such as the disjoint-paths problem, minor containment, topological minor containment, $\mathcal{F}$-topological minor deletion, and many more (see the appendix of~\cite{golovach2022model} for several examples indicating the expressibility potential of \FODP). It was already shown in~\cite{golovach2022model} that the model checking problem for disjoint-paths logic is fixed-parameter tractable on classes with excluded minors. The essence of the meta theorem of~\cite{golovach2022model} is again the irrelevant vertex technique.

\paragraph{Our results}
In this work, we prove that for every graph class excluding a fixed graph as a topological minor, the model checking problem for $\FODP$ is fixed-parameter tractable. More precisely, we prove the following result. 

\begin{theorem}\label{main_theorem}
  There is an algorithm that, given a graph $G$ (with additional vertex colors) that excludes a graph $H$ as a topological minor, and an $\FODP$ formula~$\phi(\bar x)$ (over the colored graph vocabulary), decides whether $G\models \exists \bar x\,\varphi(\bar x)$ in time $f(H,\varphi)\cdot |V(G)|^{3}$, where $f$ is a computable function depending on $H$.
  Moreover, if $G\models \varphi(\bar v)$ for some $\bar v\in V(G)^{|\bar x|}$, the algorithm outputs such a tuple $\bar v$.
\end{theorem}

Note that the algorithm needs no information about the graph $H$ that is excluded as a topological minor. 
It works uniformly over all graphs and only the running time depends on $H$. 

This essentially settles the question of tractable model checking for $\FODP$ on subgraph-closed classes, since it is already known (see \cite{SchirrmacherSV22}) that the model checking problem for the more restrictive separator logic is AW[$\star$]-hard on subgraph-closed classes that do not exclude a topological minor and admit efficient encoding. \smallskip

Beyond the direct application to classes that exclude a topological minor, \autoref{main_theorem} has further applications for the design of parameterized algorithms as follows.
Suppose that $\Pi$ is some $\FODP$-expressible parameterized problem whose instance is a graph~$G$ and some integer $k$.
Suppose also $\Pi$ satisfies the following property:
\textsl{The topological minor containment of some particular graph (whose size depends on $k$) in the input graph~$G$ directly certifies a {\sf yes}- or a {\sf no}-answer}.
Then, because of \autoref{main_theorem}, $\Pi$ admits a (parameterized) $\Oof(f(k)\cdot n^3)$ time algorithm. 

As a characteristic example of the above, we mention the \textsc{$\mathcal{F}$-Topological Minor Deletion} problem, defined for some finite collection $\mathcal{F}$ of graphs: \textsl{Given a graph $G$ and an integer $k$, the question is whether $G$ contains a set $S$ of $k$ vertices whose removal from~$G$ gives a graph excluding all graphs in $\mathcal{F}$ as topological minors}.
Fomin, Lokshtanov,  Panolan, Saurabh, and Zehavi proved in~\cite{FominLP0Z20hitti}  that \textsc{$\mathcal{F}$-Topological Minor Deletion} admits a time $\Oof(f(k)\cdot n^4)$ algorithm.
As the presence of a big enough (as a function of $k$ and $\mathcal{F}$) clique as a topological minor implies directly that the instance $(G,k)$ is a {\sf no}-instance, the problem reduces to graphs excluding some clique as a topological minor and can be solved, using \autoref{main_theorem}, in time $\Oof(f(k)\cdot n^3)$.

\paragraph{Our techniques}
Our meta theorem combines the approaches of both~\cite{golovach2022model} and~\cite{PilipczukSSTV22}. 
We start by decomposing the input graph into unbreakable parts,  using the decomposition of Cygan, Lokshtanov, Pilipczuk, Pilipczuk, and Saurabh~\cite{cygan2019minimum}. 
On each part, we distinguish two cases. We handle the first case, when a part excludes a minor, by the model checking result for classes with excluded minors~\cite{golovach2022model}. 
In the second case, when a part contains large minors, we use 
the ``\textsl{generic folio lemma}'' proved by Robertson and Seymour in~\cite{robertson1995graph} in their study of the disjoint-paths problem. 
Based on this lemma, we prove that every formula $\phi(\bar x)$ of disjoint-paths logic on unbreakable graphs with large complete graphs as minors is in fact equivalent to another first-order formula~$\psi(\bar x)$. This generalizes the results of~\cite{PilipczukSSTV22} for separator logic to disjoint-paths logic on unbreakable graphs with large clique minors. 
We then use a dynamic programming approach over tree decomposable graphs to combine the solutions of the unbreakable parts into a global solution.

Our dynamic programming approach is a variation of the standard approach for~\MSO\ or~\FO, which is based on the computation and combination of types, building on the Feferman-Vaught decomposition theorem for these logics. 
A similar decomposition also holds for disjoint-paths logic, however, when decomposing over larger separators, the number of disjoint paths that need to be queried increases. As a consequence, the decomposition theorem cannot be applied unboundedly often in the dynamic programming procedure, as we would have to deal with larger and larger quantifier ranks. 
Instead, our approach uses the ideas of~\cite{LokshtanovR0Z18} and~\cite{grohe2011finding}, where large structures are replaced by small structures of the same type. 
However, since the satisfiability problem already of plain~\FO\ is undecidable, it is not possible to find representative structures of the same type whose size is bounded by any computable function, which would lead to a non-computable dependency on $\phi$ in the running time, and hence to a non-uniform algorithm. 
However, since we can solve the model checking problem on each bag locally, we can compute a small representative of the game graph of the model checking game (enriched with colors and further vertices to ensure the same connectivity for a fixed number of disjoint paths of the substructure). 
With an appropriate rewriting of formulas, this structure can serve as a small representative structure that can be maintained over the dynamic programming procedure and we obtain a uniform fpt algorithm.

Let us also comment on why we failed to use the framework of~\cite{SchirrmacherSV22} to combine the solutions of the unbreakable parts into a global solution, but need to fall back to the dynamic programming approach. The obstacle arises from the fact that the tree decomposition into unbreakable parts (with parts of unbounded size) builds on an underlying tree with unbounded branching degree. This requires a ``simultaneous'' dynamic programming step when progressing from the children of a node to the node itself. Such a step was possible in the case of separator logic, and in fact for all properties that can be encoded as \FO~formulas with \MSO~subformulas that are essentially restricted to speak about the tree order of the tree decomposition into unbreakable parts (this logic is called $\FOMSO$ in~\cite{SchirrmacherSV22}). We were not able to combine the disjoint-paths queries simultaneously over unboundedly many child nodes of a node due to the many possibilities to route disjoint paths through the children (compare with the mentioned problems with the Feferman-Vaught decomposition theorem). This is in contrast to the comparatively simple connectivity queries of separator logic, where the solution for the children is unique and can be encoded into the torso of a bag. Note also that these problems cannot be handled by going to nice tree decompositions with an underlying binary tree. This translation requires a copying of bags (which are unboundedly large) and does not allow an encoding of the decomposition in a tree with a constant number of colors. This would be necessary such that logic can identify copies of a vertex.

\smallskip
A result that is weaker than ours follows also from the work of Lokshtanov, Ramanujan, Saurabh, and Zehavi~\cite{LokshtanovR0Z18} who proved the following result: For every $\mathsf{CMSO}$ sentence~$\phi$, if there is an $\Oof(n^d)$ algorithm (for $d>4$) to test the truth of~$\phi$ over unbreakable graphs, then there is an $\Oof(n^d)$ algorithm to test the truth of $\phi$ over all graphs. Since $\FODP$ is a fragment of MSO, our results on unbreakable graphs together with the result of~\cite{LokshtanovR0Z18} implies the \textsl{existence} of an $\Oof(n^4)$ model checking algorithm for every fixed $\FODP$ sentence~$\phi$ on every class excluding a topological minor.
The main caveat of the proof based on the result of  \cite{LokshtanovR0Z18}  is that it is non-constructive, and we can only conclude the existence of an efficient model checking algorithm. Our algorithm is fully constructive. 

\paragraph{Follow-up work}
The decomposition-and-typing machinery developed here for $\FODP$ on topological-minor-free classes
is extended and used as a central subroutine in follow-up work on $\CMSO$ with disjoint-paths predicates,
where second-order quantification is restricted to sets satisfying prescribed structural restrictions~\cite{Sau26}.

\paragraph{Organization}
\Cref{sec:prelims} collects preliminaries on (topological) minors, (unbreakable) tree decompositions, boundaried graphs, and $\FODP$.
In \Cref{sec:collapse}, we prove the collapse result showing that, on unbreakable graphs with sufficiently large clique minors, $\FODP$ formulas can be replaced by equivalent $\FO$ formulas.
\Cref{sec:dp-game-trees} develops our main technical subroutine: computing bounded-size, type-preserving representatives for boundaried instances, distinguishing the cases of large clique minors and excluding large clique minors.
In \Cref{sec:compositionality}, we establish the compositionality of (extended annotated) types under gluing, in Feferman--Vaught style.
Finally, \Cref{subsec:main-proof} combines these ingredients into the bottom-up dynamic program over a strongly unbreakable decomposition, proving \Cref{main_theorem}, and \Cref{sec:conclusion} concludes with open questions.

\section{Preliminaries}\label{sec:prelims}

Given a $k\in\mathbb{N}$, we use $[k]$ to denote the set $\{1,\ldots,k\}$.
For two $k$-tuples of sets $\bar X=(X_1,\dots,X_k)$ and $\bar Y=(Y_1,\dots,Y_k)$, we write $\bar X\subseteq \bar Y$
to mean $X_i\subseteq Y_i$ for all $i\in[k]$. 

\subsection{Graphs}
\paragraph{Graphs}
All graphs in this paper are finite, undirected graphs without loops and without multi-edges. We write $V(G)$ for the vertex set and $E(G)$ for the edge set of a graph $G$. We write $|G|$ for $|V(G)|$ and $\|G\|$ for $|V(G)|+|E(G)|$. 
For a vertex subset $X\subseteq V(G)$, we write $G[X]$ for the subgraph of $G$ induced by $X$. 
The complete graph with $t$ vertices is denoted $K_t$.

\paragraph{Disjoint paths}
Let $G$ be a graph and $u,v\in V(G)$. A $u$-$v$-path $P$ in $G$ is a sequence $v_1,\dots, v_\ell$ of pairwise different vertices such that $\{v_i,v_{i+1}\}\in E(G)$ for all $1\leq i<\ell$ and $v_1=u$ and $v_\ell=v$. The vertices $v_2,\ldots, v_{\ell-1}$ are the \emph{internal vertices} of $P$ and the vertices~$u$ and~$v$ are its \emph{endpoints}. Two vertices $u,v$ are \emph{connected} if there exists a path with endpoints $u,v$. A graph is connected if every two of its vertices are connected. 
Two paths $P,Q$ are \emph{internally vertex-disjoint} if no vertex of one path appears as an internal vertex of the other path. 
They are \emph{disjoint} if their vertex sets are disjoint. 
A set $X\subseteq V(G)$ is \emph{$m$-partition-linked} in a graph $G$ if all pairs $x_1,y_1,\ldots, x_r,y_r$ of pairwise different vertices from~$X$, $1\leq r\leq m$, are linked by pairwise disjoint paths.

\paragraph{Trees and tree orders}
An acyclic and connected graph $T$ is a \emph{tree}. 
By assigning a distinguished vertex $r$ as the root of a tree, we impose a tree order $\preceq_T$ on $V(T)$ by $x\preceq_T y$ if $x$ lies on the unique path (possibly of length~$0$) from $y$ to $r$. If $x\preceq_T y$, we call $x$ an \emph{ancestor} of $y$ in $T$. Note that by this definition, every node is an ancestor of itself. We write $y\succeq_T x$ to mean $x\preceq_T y$. We drop the subscript $T$ if it is clear from the context. 
We write $\parent(x)$ for the parent of a non-root node $x$ of $T$, and $\children(x)$ for the set of children of $x$ in~$T$. We define $\parent(r)=\bot$. 
A \emph{star} is a tree of height one.
We call the root of a star its \emph{center} and the leaves of a star its \emph{leaves}.

\paragraph{Minors and topological minors}
A graph $H$ is a \emph{minor} of a graph $G$ if $H$ can be obtained from a subgraph of $G$ by contracting edges. 
A \emph{minor model} of $H$ in $G$ is an injective mapping $\eta$ that maps vertices of~$H$ to pairwise disjoint connected subgraphs of $G$ such that for every $\{u,v\}\in E(H)$ there exist $x\in \eta(u)$ and $y\in \eta(v)$ with $\{x,y\}\in E(G)$. The set $\eta(v)$ is called the \emph{branch set} of $v$. 
A {\em{topological minor model}} of a graph~$H$ in a graph $G$ is an injective mapping $\eta$ that maps vertices of $H$ to vertices of $G$ and edges of~$H$ to pairwise internally vertex-disjoint paths in $G$ so that for every $\{u,v\}\in E(H)$ the path $\eta(\{u,v\})$ has the endpoints $\eta(u)$ and $\eta(v)$.
The vertices $\eta(v),v\in V(H)$ are called \emph{principal vertices} of the model in $G$.
A graph $H$ is a \emph{topological minor} of $G$ if there is a topological minor model of $H$ in $G$.
We call a graph $G$ {\em{$H$-minor-free}}, and {\em{$H$-topological-minor-free}}, respectively, if $H$ is not a minor, or topological minor of $G$. 
We call a class $\Cc$ of graphs {\em{(topological-)minor-free}} if there exists a graph~$H$ such that every member of $\Cc$ is $H$-(topological-)minor-free. 
A set $U\subseteq V(G)$ \emph{shatters} a $K_t$-minor of $G$ if there is a minor model of $K_t$ in $G$ such that each of its branch sets intersects $U$.

\paragraph{Tree decompositions}
A {\em{tree decomposition}} of a graph $G$ is a pair $\Tt=(T,\bag)$, where~$T$ is a rooted tree and $\bag\colon V(T)\to 2^{V(G)}$ is a mapping assigning to each node $x$ of $T$ its {\em{bag}} $\bag(x)$, which is a subset of vertices of $G$ such that the following conditions are satisfied:
\begin{enumerate}
  \item For every vertex $u\in V(G)$, the set of nodes $x\in V(T)$ with $u\in \bag(x)$ induces a non-empty and connected subtree of $T$.
  \item For every edge $\{u,v\}\in E(G)$, there exists a node $x\in V(T)$ with $\{u,v\}\subseteq \bag(x)$.
\end{enumerate}

The \emph{width} of a tree decomposition is the value $\max_{t\in V(T)}|\mathsf{bag}(t)|-1$ and the \emph{treewidth} of a graph $G$ is the minimum possible width of a tree decomposition of $G$. 

\medskip
Recall that if $r$ is the root of $T$, then we have $\parent(r)=\bot$. We define $\bag(\bot)=\emptyset$.
For a node $x\in V(T)$, we define
the {\em{adhesion}} of $x$ as
  $\adh(x)\coloneqq \bag(\parent(x))\cap \bag(x);$
  the \emph{margin} of $x$ as $\mrg(x)\coloneqq \bag(x)\setminus \adh(x)$;
  the {\em{cone at $x$}} as
  $\cone(x)\coloneqq \bigcup_{y\succeq_T x} \bag(y);$
  and the {\em{component at $x$}} as
  $\cmp(x)\coloneqq \cone(x)\setminus \adh(x).$ 
  We will see in a moment that we may assume that the component at $x$ is connected for every node $x$, which justifies the name component. 
The {\em{adhesion}} of a tree decomposition $\Tt=(T,\bag)$ is defined as the largest size of an adhesion, that is, $\max_{x\in V(T)}|\adh(x)|$.

A \emph{star decomposition} of a graph $G$ is a tree decomposition $(T,\mathsf{bag})$ where $T$ is a star.
A star decomposition $(T,\mathsf{bag})$
is \emph{tight} if for every leaf $\ell$ of $T$
\begin{enumerate}
  \item the margin $\mrg(\ell)$ is non-empty;
  \item there is no leaf $\ell'\neq \ell$ with $\mathsf{adh}(\ell) = \mathsf{adh}(\ell')$; and
  \item every vertex of $\mathsf{adh}(\ell)$ is adjacent to at least one vertex of each connected component of $G[\mathsf{mrg}(\ell)]$.
\end{enumerate}
We remark that not every star decomposition can be turned into a tight star decomposition with the same bound on adhesion sizes. 

\paragraph{Unbreakability}
A {\em{separation}} in a graph $G$ is a pair $(A,B)$ of subsets of vertices of $G$ such that $A\cup B=V(G)$ and there is no edge with one endpoint in \mbox{$A\setminus B$} and the other endpoint in $B\setminus A$.
The {\em{separator}} of $G$ associated with $(A,B)$ is the intersection \mbox{$A\cap B$} and the {\em{order}} of a separation is the size of its separator.

For $q,k\in \N$, a vertex subset $X$ in a graph $G$ is {\em{$(q,k)$-unbreakable}} if for every separation~$(A,B)$ of~$G$ of order at most $k$, we have
\[|A\cap X|\leq q\qquad\textrm{or}\qquad |B\cap X|\leq q.\]

We say that a graph $G$ is \emph{$(q,k)$-unbreakable} if $V(G)$ is $(q,k)$-unbreakable in $G$.
The notion of unbreakability can be lifted to tree decompositions by requiring it from every individual bag.

\begin{definition}
  Given $q,k\in \mathbb{N}$ and a tree decomposition $\mathcal{T}=(T,\mathsf{bag})$ of a graph $G$,
we say that a node $x\in V(T)$ has the {\em{strong $(q,k)$-unbreakability property}} if $G[\mathsf{bag}(x)]$ is $(q,k)$-unbreakable in $G[\cone(x)]$.
We say that a tree decomposition $\mathcal{T}$ is \emph{strongly $(q,k)$-unbreakable} if all its nodes have the strong $(q,k)$-unbreakability property.
\end{definition}

\begin{theorem}[\cite{cygan2019minimum}]\label{thm:strong-unbreakability}
  There is a function $q(k)\in 2^{\Oh(k)}$ such that for every graph $G$ and $k\in \N$, there exists a strongly $(q(k),k)$-unbreakable tree decomposition of $G$ of adhesion at most $q(k)$. Moreover, given~$G$ and~$k$, such a tree decomposition can be computed in time $2^{\Oh(k^2)}\cdot |G|^2\cdot \|G\|$. 
\end{theorem}

\subsection{Signatures and Logic}
\label{sec:signatures-logic}

\paragraph{Colored graph signatures}
We fix a finite relational signature $\Sigma$ of the form
$\Sigma  = \{E\}\ \cup\ \mathcal C$,
where $E$ is a binary relation symbol (which will always be interpreted as the irreflexive and symmetric edge relation) and $\mathcal C$ is a finite set of unary
relation symbols (called \emph{colors}).
A $\Sigma$-structure is called a \emph{$\Sigma$-colored graph}. 
We will often just speak of graphs and mean $\Sigma$-colored graphs for a signature $\Sigma$ that is clear from the context. 
We remark that this restriction of the signature is no loss of generality, as general relational structures can be handled by the
standard incidence encoding into vertex-colored graphs. 

\paragraph{First-order logic \textup{(\FO)}}
Formulas of $\FO[\Sigma]$ are built from atomic formulas
$x=y$, $E(x,y)$, $C(x)$ for $C\in\mathcal C$,
using the Boolean connectives $\neg,\wedge,\vee$ and the quantifiers $\exists x,\forall x$ over vertex variables.
A variable not in the scope of a quantifier is \emph{free}, and a formula without free variables is a
\emph{sentence}. We write $\phi(\bar x)$ to indicate that the free variables of $\phi$ are contained in $\bar x$.
A \emph{valuation} of $\bar x$ in $G$ is a function $\bar a:\bar x\to V(G)$; we write $V(G)^{\bar x}$ for the
set of all such valuations.
The satisfaction relation $(G,\bar a)\models \phi(\bar x)$ is defined as usual by structural induction on
$\phi$. We also write $G\models \phi(\bar a)$ and define
$\phi(G) := \{\bar a\in V(G)^{\bar x}\mid G\models \phi(\bar a)\}$.

\paragraph{Disjoint-paths queries}
Let $G$ be a graph, let $r\in\mathbb N$, and let
\mbox{$s_1,t_1,\dots,s_r,t_r\in V(G)$}.
We write $\DP_r\bigl[(s_1,t_1),\dots,(s_r,t_r)\bigr]$
to denote the predicate that is true in $G$ if and only if there exist $r$ pairwise vertex-disjoint paths
$P_1,\dots,P_r$ in $G$ such that $P_i$ has endpoints $s_i$ and~$t_i$ for every $i\in[r]$.
(We allow $s_i=t_i$, in which case $P_i$ is the trivial path; if you prefer to only consider pairwise
distinct terminals, you can instead stipulate that the predicate is false whenever the tuple
$(s_1,t_1,\dots,s_r,t_r)$ is not pairwise distinct.)

\paragraph{Disjoint-paths logic \textup{(\FODP)}}
The logic $\FODP[\Sigma]$ is obtained from $\FO[\Sigma]$ by additionally allowing atomic formulas of the form
\[
\DP_r\bigl[(x_1,y_1),\dots,(x_r,y_r)\bigr]\qquad(r\ge 1),
\]
where $x_1,y_1,\dots,x_r,y_r$ are vertex variables.
The satisfaction relation for $\FODP[\Sigma]$ is as for $\FO[\Sigma]$, with the additional clause that for a
$\Sigma$-colored graph $G$ and a valuation $\bar a$,
$(G,\bar a)\models \DP_r\bigl[(x_1,y_1),\dots,(x_r,y_r)\bigr]$
holds if and only if the underlying graph of $G$ satisfies
$\DP_r[(\bar a(x_1),\bar a(y_1)),\dots,(\bar a(x_r),\bar a(y_r))]$.

\subsection{Boundaried graphs and their (extended) folios}

Let $\Sigma$ be a colored graph signature and let $b\in\mathbb{N}.$
A  \emph{$b$-boundaried ($\Sigma$-colored) graph} is a triple $\mathbf{G} = (G,B,\rho)$ where~$G$ is a ($\Sigma$-colored) graph, $B \subseteq V(G)$, $|B| = b,$ and
$\rho\colon B \to [b]$ is a bijective function (see~\cite[Definition 12.9.1]{DowneyF13fund}).
We call $B$ the \emph{boundary} of $\mathbf{G}$, which we denote by $B(\mathbf{G})$, and the vertices of $B$ \emph{the boundary vertices} 
of~$\mathbf{G}$. We also call $G$ \emph{the underlying  graph} of $\mathbf{G}$.
A \emph{boundaried graph} is a $b$-boundaried graph, for some~$b\in\mathbb{N}.$
We remark that from the logical point of view, we see $b$-boundaried $\Sigma$-colored graphs as $\Sigma^+$-colored graphs, where $\Sigma^+$ is obtained by enhancing $\Sigma$ with extra $b$ colors, which are bijectively mapped to the boundary vertices.

As in~\cite{RobertsonS95XIII} (see also~\cite{BasteST20acom}), 
we define the \emph{detail} of $\mathbf{G}$ as  
\[\mathsf{detail}(\mathbf{G}):=\max\{|E(G)|,|V(G)\setminus B|\}.\]

Two $b$-boundaried graphs $\mathbf{G}_{1}=(G_{1},B_{1},\rho_{1})$ and
$\mathbf{G}_{2}=(G_{2},B_{2},\rho_{2})$ are \emph{isomorphic} if $G_{1}$ is isomorphic to $G_{2}$ 
via a bijection $\phi\colon  V(G_{1})\to V(G_{2})$ such that $\rho_{1}=\rho_{2}\circ\phi|_{B_{1}} ,$ i.e.,~the vertices of $B_{1}$ are mapped via $\phi$ to equally indexed vertices  of $B_{2}.$
We denote by~${\mathcal{B}_b}$ the set of all (pairwise non-isomorphic)  $b$-boundaried graphs, and we set $\mathcal{B} := \bigcup_{b\in\N} \mathcal{B}_b$ for the class of all boundaried graphs.

\paragraph{Compatibility}
Let $\mathbf{G}_{1}=(G_1,B_1,\rho_1)$ and $\mathbf{G}_{2}=(G_2,B_2,\rho_2)$ be two $b$-boundaried $\Sigma$-colored graphs.
We say that $\mathbf{G}_{1}$ and $\mathbf{G}_2$ are \emph{compatible} if $\rho_{2}^{-1}\circ \rho_{1}$  is an isomorphism from $G_{1}[B_{1}]$ to $G_{2}[B_{2}]$ such that for every color $C\in\Sigma$ and every $i\in[b]$, $\rho^{-1}_1(i)$ belongs to the interpretation of $C$ in $G_1$ if and only if $\rho^{-1}_2(i)$ belongs to the interpretation of $C$ in $G_2$.

\paragraph{Topological minors of boundaried graphs}
We say that a boundaried graph $\mathbf{G}'=(G',B',\rho')$ is a (boundaried) \emph{subgraph} of $\mathbf{G}$ if
$G'$ is a subgraph of $G$, $B'=B$, and $\rho'=\rho$. 

Let $\textbf{M}=(M,B,\rho)$ be a boundaried graph and let $T\subseteq V(M)$ with $B\subseteq T,$ such that 
every vertex in $V(M)\setminus T$ has degree two.
We  call  $(\textbf{M},T)$ a \emph{\textsf{btm}-pair}
and we  define  $\mathsf{diss}(\textbf{M},T)=(\mathsf{diss}(M, T),B,\rho),$ where $\mathsf{diss}(M, T)$ is the result of the dissolution\footnote{The \emph{dissolution} of a vertex $v$ of degree two 
whose two neighbors are $x$ and $y$ 
is the operation of removing~$v$ and adding the edge $\{x,y\}$ (if it does not already exist).} in $M$ of all vertices in~$V(M)\setminus T$.
We refer to the set $T$ as the set of \emph{branch vertices} of $(\mathbf{M},T)$.
Note that we do not permit dissolution of boundary vertices, as we demand all of them to be branch vertices. 
If $\mathbf{G}=(G,B,\rho)$ is a boundaried graph, we say that $({\bf M},T)$ is an \emph{\textsf{btm}-pair of $\mathbf{G}$} if it is an \textsf{btm}-pair  where $\mathbf{M}$ is a subgraph of $\mathbf{G}$ and $B\subseteq T$. 
We say that a boundaried graph $\textbf{G}_{1}$ is an \emph{topological minor}
of a boundaried graph~$\mathbf{G}_{2},$ denoted by $\mathbf{G}_{1}\preceq\mathbf{G}_{2},$ if $\mathbf{G}_{1}$ and~$\mathbf{G}_{2}$ are compatible and~$\mathbf{G}_{2}$ has a \textsf{btm}-pair $(\textbf{M},T)$
such that  $\mathsf{diss}(\textbf{M},T)$ is isomorphic to $\textbf{G}_{1}.$

\paragraph{Folios and extended folios}
Given a boundaried graph $\mathbf{G}$ and $\delta\in\mathbb{N}$, we define the \emph{$\delta$-folio} of $\mathbf{G}$ as
\[
\delta\mathsf{\text{-}folio}(\mathbf{G}) := \bigl\{\, \mathbf{G}'\in\mathcal{B} \ \bigm|\ \mathbf{G}'\preceq \mathbf{G}
\text{ and }\mathsf{detail}(\mathbf{G}')\le \delta \,\bigr\}.
\]

We next give the definition of \emph{extended folios}, originating from~\cite{grohe2011finding}.
Let $\mathbf{G} = (G,B,\rho)$ be a boundaried graph.
Recall that, given $I\subseteq \binom{[|B|]}{2}$, we write~$\mathbf{G}^I$ 
for the boundaried graph obtained from $\mathbf{G}$ by adding 
in its underlying graph the edges in 
\[
\bigl\{\, \{\rho^{-1}(i),\rho^{-1}(j)\} \ \bigm|\ \{i,j\}\in I \,\bigr\}.\] 

The \emph{extended  $\delta$-folio of $\mathbf{G} = (G,B,\rho)$}
is
\begin{eqnarray*}
\delta\mathsf{\text{-}ext\text{-}folio}(\mathbf{G})& = & \{(I,\delta\mathsf{\text{-}folio}(\textbf{G}^I))\mid 
I\in \binom{[|B|]}{2}\}. 
\label{eq_ext_f}
\end{eqnarray*}

We say that two boundaried graphs are \emph{$\delta$-equivalent} if they are compatible and they have the same extended $\delta$-folio.
We will use the following result from~\cite{grohe2011finding}.

\begin{proposition}[Lemma 2.2 and Theorem 3.1 of~\cite{grohe2011finding}]\label{prop:bound-rep-folio}
There is a computable function $f_{\mathsf{folio}}\colon \mathbb{N}^2\to\mathbb{N}$ such that
for every $\delta,b\in\mathbb{N}$ and every extended $\delta$-folio $\mathcal{F}$, the following holds: Every minimum-size (in terms of vertices) $b$-boundaried graph whose extended $\delta$-folio is~$\mathcal{F}$ has size at most $f_{\mathsf{folio}}(\delta,b)$.
Moreover, there is an algorithm that, given $\delta,b,$ and an $n$-vertex $b$-boundaried graph $\mathbf{G}$, computes, in time $\mathcal{O}_{\delta,b}(n^3)$, a $b$-boundaried graph~$\mathbf{H}$ of size at most $f_{\mathsf{folio}}(\delta,b)$ such that $\mathbf{G}$ and $\mathbf{H}$ are $\delta$-equivalent.
\end{proposition}

\section{Collapse of \texorpdfstring{$\FODP$}{FO+DP} on unbreakable graphs with large clique minors}\label{sec:collapse}

We now show that $\FODP$ collapses to plain first-order logic on appropriately unbreakable graphs that contain large clique minors. 
We use the following ``\textsl{generic folio lemma}'' of Robertson and Seymour from~\cite{robertson1995graph}.

\begin{proposition}[Lemma 5.4 of~\cite{robertson1995graph}]\label{lem:RobSey}
  Let $G$ be a graph and $Z\subseteq V(G)$. Let $t\geq \frac{3}{2}|Z|$ and let $B_1,\ldots, B_t$ be the branch sets of a $K_t$ minor of $G$. Suppose that there is no separation $(G_1,G_2)$ of $G$ of order $<|Z|$ with $Z\subseteq V(G_1)$ and $B_i\cap V(G_1)=\emptyset$ for some $i\in \{1,\ldots, t\}$. 
  Then for every partition $(Z_1,\ldots, Z_m)$ of $Z$ into non-empty subsets, there are pairwise disjoint connected subgraphs $T_1,\ldots, T_m\subseteq G$ such that $V(T_i)\cap Z=Z_i$ for all $i\in \{1,\ldots, m\}$. 
\end{proposition}

Note that in particular under the conditions of the \cref{lem:RobSey}, 
the set $Z$ is $m$-partition-linked. 

In our dynamic programming approach, we will proceed bottom-up and hence will be handling the situation of a bag with its children. 
For this reason, we formulate the collapse result in a slightly more general form and not just for unbreakable graphs.
The situation that we will establish as an invariant in the dynamic program is that of a tight star decomposition. 
We first extract the following combinatorial lemma, which easily implies the claimed collapse result.

\begin{lemma}\label{lem:combinatorial-collapse-unbreakable}
    For all $q,r,c,d\in\mathbb{N}$, there are $k:=k(r)$, $t:=t(q,r)$, $L:=L(q,r,c,d)$, $s:=s(q,r,c,d)\in\mathbb{N}$ 
    such that the following holds.
    Let $G$ be a graph and let $(T,\mathsf{bag})$ be a tight star decomposition of $G$ such that
    \begin{itemize}[nosep]
        \item the adhesion of $(T,\mathsf{bag})$ is at most $d$;
         \item if $s$ is the center of $T$, then $\mathsf{bag}(s)$ is $(q,k)$-unbreakable in $G$ and shatters a $K_t$-minor of $G$; and
        \item $|\mathsf{bag}(\ell)|\le c$ for every leaf $\ell$ of $T$.
    \end{itemize}
   Then for every $U=\{x_1,y_1,\ldots,x_r,y_r\}\subseteq V(G)$,
   \begin{itemize}
    \item there is a (unique) partition $\mathcal{C}$ of $U$ such that two vertices $v,u\in U$ are in the same part $C\in\mathcal{C}$ if there is a sequence of vertices $z_1,\ldots, z_\beta\in U$ such that $v=z_1$ and $u=z_\beta$ and for every $i\in[\beta-1]$, the distance (in $G$) between $z_i$ and $z_{i+1}$ is at most~$L$, such that
    \item for every $C\in\mathcal{C}$, there is a separation $(X_C,Y_C)$ of $G$ with 
    \begin{enumerate}[nosep]
      \item $C\subseteq X_C$, 
      \item $|X_C\cap Y_C|\leq k$, and
      \item $|X_C|\leq k \cdot (\max\{\binom{q}{\le d},k\}\cdot c + q)$, and
    \end{enumerate}
    for any choice of such separations, writing $S_C:=X_C\cap Y_C$, $S:=\bigcup_{C\in\mathcal{C}}S_C$,
and letting $D$ be the union of the connected components of $G[X_C]$ that intersect $C$ (over all~$C\in\mathcal{C}$),
we have:
\begin{enumerate}[nosep]
  \item $|S|\le k^2$,
  \item $|D|\le k^2\cdot(\max\{\binom{q}{\le d},k\}\cdot c+q)$, and
  \item $S$ is $r$-partition-linked in $G\setminus(D\setminus S)$.
\end{enumerate}
\end{itemize}
\end{lemma}

\begin{figure}[ht]
  \centering
  \begin{tikzpicture}
    \def\G{(0,0) ellipse (4.5 and 2.5)}
    \def\D{(3.4,1.7) ellipse (1.5 and 1)}
    \def\Dd{(4.3,.8) ellipse (1.4 and .9)}
    \def\Dr{(3.4,-1.7) ellipse (1.5 and 1)}
    
    \def\C{(3.2,2.7) ellipse (.55 and .55)}
    \def\Cc{(4.9,.9) ellipse (.45 and .45)}
    \def\Cr{(3,-2.7) ellipse (.55 and .55)}

    \def\Kt{(-2,0) ellipse (1.5 and 1.5)}

    \draw \G;

    \draw[rotate=15] \D;
    \draw \Dd;
    \draw[rotate=-15] \Dr;
    
    \draw[blue] \C;
    \draw[blue] \Cc;
    \draw[blue] \Cr;
    
    \draw \Kt;
    
    \begin{scope}
      \clip \G;
      \draw[rotate=15,red] \D;
      \draw[red] \Dd;
      \draw[rotate=-15,red] \Dr;
    \end{scope}

    \begin{scope}
      \clip[rotate=15] \D;
      \draw[red] \G;
    \end{scope}

    \begin{scope}
      \clip \Dd;
      \draw[red] \G;
    \end{scope}

    \begin{scope}
      \clip[rotate=-15] \Dr;
      \draw[red] \G;
    \end{scope}
    
    \node (x1) [scale=.75] at (3.5,2.9) {$x_1$};
    \node (x7) [scale=.75] at (2.9,2.7) {$x_7$};
    \node (y3) [scale=.75] at (3.3,2.5) {$y_3$};

    \node (x2) [scale=.75] at (5,.8) {$x_2$};

    \node (x9) [scale=.75] at (3,-2.5) {$x_9$};
    \node (y1) [scale=.75] at (2.8,-2.7) {$y_1$};
    \node (y5) [scale=.75] at (3.3,-2.8) {$y_5$};

    \node (C1) [blue,scale=.75] at (3.1,2.98) {$C_1$};
    \node (C2) [blue,scale=.75] at (4.77,1.05) {$C_2$};
    \node (Cr) [blue,scale=.75] at (3,-3) {$C_r$};
  
    \node (D1) [scale=.75] at (2.2,2.8) {$D_{C_1}$};
    \node (D2) [scale=.75] at (4.2,1.3) {$D_{C_2}$};
    \node (Dr) [scale=.75] at (2.1,-2.9) {$D_{C_r}$};

    \node (S1) [red,scale=.75] at (2.2,1.9) {$S_{C_1}$};
    \node (S2) [red,scale=.75] at (3.7,.7) {$S_{C_2}$};
    \node (Sr) [red,scale=.75] at (2.3,-1.9) {$S_{C_r}$};

    \node (Kt) at (-2,0) {$K_t$};

    \draw[blue,dotted,line width=1.2] (3.6,2.75) to[out=-45,in=0] (2.8,1.7);
    \draw[blue,dotted,line width=1.2] (2.8,2.6) to[bend right=5] (1.8,2);
    \draw[blue,dotted,line width=1.2] (3.3,2.35) to[bend left=5] (2.6,1.85);

    \draw[blue,dotted,line width=1.2] (4.85,.7) to[bend left=5] (4.05,.5);

    \draw[blue,dotted,line width=1.2] (2.8,-2.45) to[bend left=10] (2.6,-1.75);
    \draw[blue,dotted,line width=1.2] (2.6,-2.7) to[out=180,in=-75] (1.7,-2);
    \draw[blue,dotted,line width=1.2] (3.3,-2.65) to[out=90,in=-50] (2.9,-1.65);
  
    \node (dots) [scale=1,rotate=-25] at (4.5,-1) {$\mathbf{\vdots}$};
  \end{tikzpicture}
  \caption{For a tight star decomposition of a graph $G$, we depict the concepts defined in \cref{lem:combinatorial-collapse-unbreakable}. 
  Here, the partition is $\mathcal{C}:=\{\{x_1,x_7,y_3\},\{x_2\},\ldots,\{x_9,y_1,y_5\}\}$,
  and for each $C\in\mathcal{C}$ and a separation~$(X_C,Y_C)$, we represent $S_C=X_{C}\cap Y_{C}$ with $X_C$ being the small part, we have $D_C$ as the union of the (small) connected components of $G\setminus S_C$ that contain a vertex of $C$.
  }
  \label{fig:dp-queries-unbreakable}
\end{figure}

Now it will be easy to show that under the conditions of \cref{lem:combinatorial-collapse-unbreakable}, there exists an \FO-formula such that 
for all $x_1,y_1,\ldots,x_r,y_r\in V(G)$, we have
\[G\models\DP_r[(x_1,y_1),\ldots, (x_r,y_r)] \quad \text{if and only if} \quad G\models \phi_{\mathsf{dp}}(x_1,y_1,\ldots,x_r,y_r),\]
that is, the disjoint-paths predicate collapses to a plain \FO-formula. 
Let us first prove the lemma. 

\begin{proof}
We set $k:=2r$ and $t:=\max\{q+1,6r^2\}$.
We also set $L:=k \cdot (\max\{\binom{q}{\le d},k\}\cdot c + q)$, where $\binom{q}{\le d}$ denotes the number of subsets of an $n$-element set of size at most $d$. 

We use $B_1,\ldots,B_t\subseteq V(G)$ to denote the branch sets of a minor model of $K_t$ in $G$.

We first consider the partition $\mathcal{C}$ of the set~$U$ such that two vertices $v,u\in U$ are in the same part $C\in\mathcal{C}$ if 
there is a sequence of vertices $z_1,\ldots, z_\beta\in U$ such that $v=z_1$ and $u=z_\beta$ and for every $i\in[\beta-1]$, the distance (in $G$) between $z_i$ and $z_{i+1}$ is at most~$L$.

Note that this partition is unique. 
For every $C\in\mathcal{C}$,
we consider a separation $(X_C,Y_C)$ of~$G$ such that the following conditions are satisfied:

\smallskip
\begin{enumerate}[nosep]
\item $C\subseteq X_C$,
\item there is a branch set $B_j$ for $j\in\{1,\ldots,t\}$ such that $B_j\subseteq Y_C\setminus X_C$, and
\item $|X_C\cap Y_C|$ is minimum possible.
\end{enumerate}

\smallskip
Let $S_C = X_C\cap Y_C$ and let $D_C$ be the union of the connected components of $G[X_C]$ that intersect $C$. 
Since $(C,V(G))$ satisfies the conditions 1-2 above, we get that $|S_C|\leq 2r=k$.

\begin{claim}\label{claim:sizeofD}
    Let $C$ be a subset of $V(G)$ with $|C|\le k$ and let $(X_C,Y_C)$ be a separation of~$G$ such that the following conditions are satisfied:
    \smallskip

    \begin{enumerate}[nosep]
        \item $C\subseteq X_C$,
        \item there is a branch set $B_j$ for $j\in\{1,\ldots,t\}$ with $B_j\cap X_C=\emptyset$, and
        \item $|X_C\cap Y_C|$ is minimum possible.
    \end{enumerate}

    \smallskip
    Also, let $D_C$ be the union of all connected components of $G[X_C]$ that intersect $C$. Then, the size of~$D_C$ is at most $k\cdot\left(\max\{\binom{q}{\le d},k\}\cdot c +q\right).$
\end{claim}
\begin{claimproof}
    We set $S=X_C\cap Y_C$. 
    We first show that because of unbreakability, the part of~$D_C$ in $\bag(s)$ is bounded, i.e.,
    \begin{equation}
        |D_C\cap \bag(s)|\le q. \label{eq1}
    \end{equation}

    To show~\eqref{eq1}, we argue as follows.
    We know that the separation $(C,V(G))$ satisfies the conditions of the claim and therefore
$|S|\leq |C|\le k$.
Since $B_j\subseteq Y_C\setminus X_C$, we have that all branch sets $B_1,\ldots,B_t$ are intersecting $Y_C$. Also, all of them intersect $\mathsf{bag}(s)$, therefore $|Y_C\cap \bag(s)|\geq q+1$.
Thus, by $(q,k)$-unbreakability of $\bag(s)$ in $G$, and since $|S|\le k$,
we have that $|X_C\cap \bag(s)|\le q$, which in turn implies that $|D_C\cap\bag(s)|\le q$.

Using~\eqref{eq1}, we next bound the size of each connected component of $D_C$. In fact, we show that for every connected component $W$ of $G[D_C]$,
\begin{equation}
    |W|\le \max\{\binom{q}{\le d},k\}\cdot c +q. \label{eq2}
\end{equation}

To show~\eqref{eq2}, we argue as follows.
First, we say that a leaf $\ell$ of $T$ is \emph{$W$-contaminated} if~$W$ contains a vertex of $\bag(\ell)$ that is not in $\bag(s)$, i.e., $W\cap \mrg(\ell)\neq \emptyset$. Our goal is to bound the number of $W$-contaminated leaves. We also say that a leaf $\ell$ of $T$ is 
\begin{itemize}
  \item \emph{$W$-rich} if $\adh(\ell)\subseteq W$,
  \item \emph{$W$-poor} if $\emptyset\subset\adh(\ell)\cap W \subset \adh(\ell)$, and
  \item \emph{$W$-empty} if $\adh(\ell)\cap W = \emptyset$.
\end{itemize}
Observe that since $W$ is connected, if a $W$-contaminated leaf $\ell$ is $W$-empty, then $W\subseteq \mrg(\ell)$ and therefore $|W|\le |\bag(\ell)|\le c$.
Therefore, it remains to examine the case where no $W$-contaminated leaf is $W$-empty.
In this case, every $W$-contaminated leaf is either $W$-rich or $W$-poor:
\begin{itemize}
    \item The number of different $W$-rich leaves is at most $\binom{q}{\le d}$; to see this note that because of the tightness of the star decomposition, for every set $A\subseteq \bag(s)$ there is at most one leaf $\ell$ with $\adh(\ell)=A$, which implies that the number of different $W$-rich leaves is upper-bounded by the number of different subsets of $W\cap \bag(s)$ of size at most~$d$, which is at most $\binom{q}{\le d}$.
    
    \item The number of $W$-contaminated $W$-poor leaves is at most $k$; to see this, let $\ell$ be such a leaf. The fact that it is $W$-poor implies the existence of a vertex $v\in \adh(\ell)$ that is not in $W$.
    Pick also a vertex $u\in W\cap \mrg(\ell)$; this exists since $\ell$ is $W$-contaminated. 
    Since the star decomposition is tight, by Property~3 there is a path $P$ in $G[\bag(\ell)]$
    connecting $u$ and $v$ whose internal vertices avoid $\adh(\ell)$ (equivalently, the internal vertices of $P$ lie in $\mrg(\ell)$).
    Note that $u\in W\subseteq X_C$, while $v\notin W$ implies $v\notin X_C$: otherwise, $u$ and $v$
    would be connected inside $G[X_C]$ and hence, $v$ would lie in the component $W$.
    Thus, $v\in Y_C\setminus X_C$, and since $(X_C,Y_C)$ is a separation, the path~$P$ must contain
    a vertex of $S=X_C\cap Y_C$. 
    Choosing the first vertex of $P$ that lies in $S$, it is not in
    $\adh(\ell)$ (as $P$ has no internal vertices in $\adh(\ell)$), hence it lies in $\mrg(\ell)$.
    Therefore, $P$ contains a vertex of $S\cap \mrg(\ell)$.
    Thus, for each $W$-contaminated $W$-poor leaf $\ell$, there should be a vertex in $\mrg(\ell)$ that is also in $S$. Since $|S|\le k$, we have that there are at most $k$ such leaves.    
\end{itemize}

Therefore, the number of $W$-contaminated leaves is at most $ \max\{\binom{q}{\le d},k\}$. Since for every leaf $\ell$, we have that $|\bag(\ell)|\le c$, we get that  $$|W\setminus \bag(s)|\le \sum_{\ell:\ \text{$W$-contaminated}}|\bag(\ell)| \le \max\{\binom{q}{\le d},k\}\cdot c.$$ Combined with~\eqref{eq1}, we get the claimed bound on the size of $W$, which shows~\eqref{eq2}.

\smallskip
To get the bound on $D_C$, note that each connected component of $D_C$ intersects $C$. Thus,
\begin{align*}
    |D_C| & \le \sum_{W:\ \text{conn. comp. of $G[D_C]$}} |W|\\
    & \le |C| \cdot \max\{|W|: W \text{ is a connected component of $G[D_C]$}\}\\
    & \le |C|\cdot\left(\max\{\binom{q}{\le d},k\}\cdot c +q\right),
\end{align*}
where the last inequality follows from~\eqref{eq2}.
\end{claimproof}

We apply \cref{claim:sizeofD} for the set $C$ and the separation $(X_C,Y_C)$. Since $|C|\le 2r=k$ and $(X_C,Y_C)$ satisfies the conditions of the claim, we get that
$$|D_C|\le k\cdot\left(\max\{\binom{q}{\le d},k\}\cdot c +q\right).$$

We set
$D:=\bigcup_{C\in\mathcal{C}} D_C$ and
$S:=\bigcup_{C\in\mathcal{C}} S_C$.
Then, we immediately get $$|D|\leq k^2 \cdot (\max\{\binom{q}{\le d},k\}\cdot c + q)$$ and $|S|\leq k^2$.
We set
$G_1:=G[D]$, and
$G_2:= G\setminus(D\setminus S)$ and prove the following claim.

\begin{claim}\label{claim:generic}
  $S$ is $r$-partition-linked in $G_2$.    
\end{claim}
\begin{claimproof}
We will use \cref{lem:RobSey} to show this.
Suppose, towards a contradiction, that there is a separation of $G_2$ that violates the conditions of \cref{lem:RobSey}.
Among all such separations of $G_2$, choose a separation $(F_1,F_2)$ such that $|F_1\cap F_2|$ is minimum possible and every connected component of $G[F_1]$ intersects $S$.
Keep in mind that for this separation, it also holds that
$|F_1\cap F_2|<|S|$, $S\subseteq F_1$, and there is a $j\in\{1,\ldots,t\}$ such that $B_j\cap F_1=\emptyset$. 
Also, note that the condition that $\frac{3}{2} |S|\le t$, which is required in \cref{lem:RobSey}, holds for the chosen $t$. 

Let $S':=F_1\cap F_2$ and $F_1':= F_1 \cup D$.
Note that $(F_1',F_2)$ is a separation of $G$
such that $F_1'\cap F_2 =S'$ and $F_1\subseteq F_1'$.
Also, it holds that every vertex of every $C\in\mathcal{C}$ is in $F_1'$,
and for some $i\in\{1,\ldots,t\}$, $B_i\cap F_1'=\emptyset$.
Because of \cref{claim:sizeofD} applied for $R=\bigcup \mathcal{C}$, we get that $|F_1'|\le k \cdot (\max\{\binom{q}{\le d},k\}\cdot c + q) = L$.

For every $C\in\mathcal{C}$, we set $S_{C}'$ to be the set of vertices of $S'$ that are connected with some path with some vertex of $S_C$ in $G[F_1']$.
We want to prove that for every two distinct $C,C'\in \mathcal{C}$, $S_{C}'$ and $S_{C'}'$ are disjoint.
Indeed, if there is a vertex $v\in S_{C}'\cap S_{C'}'$, then there is a path $P$ in $G[F_1']$ connecting vertices of $C$ and $C'$, which contains $v$.
By the definition of $\mathcal{C}$, every two vertices $z\in C$ and $z'\in C'$ are in distance more than $L$ in~$G$; thus, $P$ should have a length more than $L$.
The fact that $V(P)\subseteq F_1'$ implies that $|F_1'|>L$, a contradiction.
Therefore, for every two distinct $C,C'\in \mathcal{C}$, the sets $S_{C}'$ and $S_{C'}'$ are disjoint.
This implies that if $|S'|<|S|$, then there is some $C\in\mathcal{C}$ such that $|S_C'|<|S_C|$.
Now notice that~$S_C'$ is separating~$C$ from at least one branch set of the minor model of~$K_t$ and the fact that $|S_C'|<|S_C|$ contradicts the minimality of $S_C$.
Therefore, the conditions of \cref{lem:RobSey} are satisfied, and from this we get that $S$ is $r$-partition-linked in $G_2$.
\end{claimproof}

This concludes the proof of \cref{lem:combinatorial-collapse-unbreakable}.
\end{proof}

Due to~\cref{lem:combinatorial-collapse-unbreakable}, we conclude that, under certain unbreakability assumptions, the disjoint-paths predicate can be expressed in $\FO$. The following result will be applied for $d=k$, which does not cause a problem, since $k$ depends only on $r$.

\begin{corollary}\label{lem:collapse-unbreakable}
    For all $q,r,c,d\in\mathbb{N}$, there are $k:=k(r),t:=t(q,r)\in\mathbb{N}$
    and an $\FO$-formula $\phi_{\mathsf{dp}}(x_1,y_1,\ldots,x_r,y_r)$ such that the following holds.
    Let $G$ be a graph and let $(T,\mathsf{bag})$ be a 
    tight star decomposition of $G$ such that
    \begin{itemize}
        \item the adhesion of $(T,\mathsf{bag})$ is at most $d$;
         \item if $s$ is the center of $T$, then $\mathsf{bag}(s)$ is $(q,k)$-unbreakable in $G$ and shatters a $K_t$-minor of $G$;
        \item $|\mathsf{bag}(\ell)|\le c$ for every leaf $\ell$ of $T$.
    \end{itemize}
   Then for all $x_1,y_1,\ldots,x_r,y_r\in V(G)$, we have that
\[G\models\DP_r[(x_1,y_1),\ldots, (x_r,y_r)] \quad \text{if and only if} \quad G\models \phi_{\mathsf{dp}}(x_1,y_1,\ldots,x_r,y_r).\]
\end{corollary}

\begin{proof}
  We will use the notation from that lemma and fix sets as defined there. 
  That is, we have our partition $\mathcal{C}$ of the set $U=\{x_1,y_1,\ldots,x_r,y_r\}$  such that for every $C\in\mathcal{C}$, there is a minimal separator $S_C$ for which the part that contains $C$ (i.e., the part $D_C$) has size at most $q$. 
  Observe that the sets $S_C$ and resulting connected components are possibly not uniquely determined, however, as stated in the lemma such sets $S_C$ satisfying the conditions exists, and any choice of such $S_C$ has the desired properties. 
  Recall that we defined $G_1:=G[D]$, and $G_2:= G\setminus(D\setminus S)$. 
  Since all numbers in \cref{lem:combinatorial-collapse-unbreakable} are fixed constants, we can existentially quantify and define all objects from the lemma. 

  Let us now assume there are disjoint paths $P_1,\ldots, P_r$ between the $x_i, y_i$ and let us describe their structure. 
  We distinguish two cases. First, if a path $P_i$ does not contain a vertex of $S$, then $P_i$ lies fully in the connected component of $G-S$ containing $x_i$ and $y_i$, and has length at most $L$. 
  Second, if a path $P_i$ contains a vertex of $S$, we can partition it into parts $Q_{x_i}$, $Q_i$, and $Q_{y_i}$, where $Q_{x_i}$ leads from $x_i$ to some vertex $z_i$ of $S$ and contains no other vertex of $S$, $Q_{y_i}$ leads from some vertex $z'_i$ of $S$ to $y_i$ and contains no other vertex of $S$, and $Q_i$ is $P_i-Q_{x_i}-Q_{y_i}$. 
  Because $S$ is $r$-partition-linked in $G_2$, we may assume that $Q_i$ uses no other vertices from the connected components containing any~$C\in \mathcal{C}$. 
  The path segments $Q_{x_i}$ and $Q_{y_i}$ have length at most $L$, while the existence of the path segment $Q_i$ is guaranteed by \cref{lem:combinatorial-collapse-unbreakable}.

  The formula $\phi$ can hence (after stating the existence of the clusters of $\mathcal C$, the separators~$S_C$ and the small sides $X_C$) express the existence of the short paths segments. With a big disjunction, we distinguish all possibilities of forming the paths according to the above structure. 

It follows by construction that 

\[G\models \phi_{\mathsf{dp}}(x_1,y_1,\ldots,x_r,y_r) \Leftrightarrow G\models\DP_r[(x_1,y_1),\ldots, (x_r,y_r)].\qedhere\]
\end{proof}

\begin{corollary}
  Let $q,r\in \N$ and let $G$ be a connected $(q,2r)$-unbreakable graph containing $K_t$ as a minor, where $t\geq \max\{q+1,6r^2\}$. Then, there is an $\FO$-formula~$\phi_{\mathsf{dp}}$, whose length depends only on $r$ and $q$, such that for all $x_1,y_1,\ldots,x_r,y_r\in V(G)$, we have that
  \[G\models\DP_r[(x_1,y_1),\ldots,(x_r,y_r)] \quad \text{if and only if} \quad G\models \phi_{\mathsf{dp}}(x_1,y_1,\ldots,x_r,y_r).\]
\end{corollary}

\section{Computing type-preserving representatives}
\label{sec:dp-game-trees}
A standard way to approach the model checking problem is via computing the $q$-type of the graph, where $q$ is the number of quantifiers of the given formula $\phi$.
The $q$-type of a graph, with respect to some logic $\mathcal{L}$ (here $\mathcal{L}$ will be $\FODP$), is basically the set of all sentences of $\mathcal{L}$ (up to logical equivalence) of at most $q$ quantifiers that the graph satisfies. Therefore, once we know the $q$-type of the input graph, we can directly decide whether it satisfies any given formula of at most $q$ quantifiers. Our strategy to compute the type of the graph is to perform dynamic-programming over a given unbreakable tree decomposition. The main subroutine of this algorithm is to compute, for each node $t$ of the tree of the tree decomposition, a new bounded-size graph that has the same $q$-type as the graph induced by $\cone(t)$ (i.e., the union of the bags of $t$ and all its descendants). However, as already noted in~\cite{golovach2022model}, the fact that $\FODP$ allows for paths of arbitrary length, does not allow to just discard vertices of the same type to get this bounded-size representative (as one would do when dealing with $\FO$). Therefore, we use \emph{annotated $q$-types} (as in~\cite{golovach2022model}) which are defined over graphs extended with a tuple of vertex subsets $R_1,\ldots,R_q$, where the $i$th variable in each formula of the $q$-type is asked to be interpreted in $R_i$. Our overall approach therefore consists in finding bounded-size representatives of the same \emph{annotated type}.
Before presenting the main result of this section, we proceed to the formal definition of annotated types and related notions.

\paragraph{Annotated types of our logic}
Let $\Sigma$ be a colored graph signature.
We say that a sentence~$\varphi$ of $\FODP[\Sigma]$ is in \emph{prenex normal form} if it is of the form
\[Q_1 x_1\ldots Q_{r}x_r\ \psi(\bar{x}),\]
where $r\in\mathbb{N}$, for every $i\in[r]$ $Q_i \in \{\forall,\exists\}$, and
$\psi(\bar{x})$ is a quantifier-free formula of $\FODP[\Sigma]$.
We refer to $r$ as the \emph{quantifier rank} of $\varphi$. We call $x_i$ the \emph{$i$th variable} of $\varphi$.

Let $r\in\mathbb{N}$. For a given $\Sigma$-structure $G$ and a tuple $\bar{R}=(R_1,\ldots,R_r)$ of $r$ subsets of~$V(G)$, we define the \emph{rank-$r$ annotated type} of $(G,\bar{R})$ with respect to $\FODP$, which we denote by $\mathsf{type}_r(G,\bar{R})$, to be the set of all (up to logical equivalence) sentences $\varphi$ of $\FODP[\Sigma]$ in prenex normal form that have quantifier rank at most $r$ and are satisfied in~$G$ when the $i$th variable of $\varphi$ is interpreted as an element of $R_i$.

During our dynamic-programming algorithm, we will be dealing with boundaried graphs, whose boundary corresponds to the adhesion of the current node in the decomposition. Therefore, the number of boundary vertices will be at most $k$, where $k$ is as in~\cref{lem:collapse-unbreakable}. As we mentioned before, the initial task is to compute a small-sized (boundaried) graph that has the same annotated type as the boundaried graph under consideration. However, it turns out that preserving the annotated type is not enough. 
A path (from the disjoint-paths predicate) may cross the boundary several times, yielding a blow-up in the number of considered paths (i.e., making $r$ depend on $k$ in~\cref{lem:collapse-unbreakable}) which we cannot afford. To handle this, we simulate possible connections between boundary vertices in the currently ``unexplored" graph (i.e., the graph without the bags of the descendants of the current node). 
This is done by adding all different combinations of edges between boundary vertices and considering a tuple of annotated types, one for each extension of graph with respect to each such combination of edges. This is captured by the notion of extended annotated types defined below. This is inspired by the notion of extended folios in~\cite{grohe2011finding} (see \cref{sec:subroutines}).

\paragraph{Extended annotated types}
Let $\mathbf{G} = (G,B,\rho)$ be a boundaried graph.
Recall that, given $I\subseteq \binom{[|B|]}{2}$, we write~$\mathbf{G}^I$ 
for the boundaried graph obtained from $\mathbf{G}$ by adding 
in its underlying graph the edges in 
\[
\bigl\{\, \{\rho^{-1}(i),\rho^{-1}(j)\} \ \bigm|\ \{i,j\}\in I \,\bigr\}.
\]
The \emph{extended rank-$r$ annotated type} of $(\mathbf{G},\bar{R})$
is the set
\begin{eqnarray*}
\mathsf{ext\text{-}type}_{r}(\mathbf{G},\bar{R})& = & \{(I,\mathsf{type}_{r}(\textbf{G}^I,\bar{R}))\mid 
I\in \binom{[|B|]}{2}\}. 
\label{eq_ext_type}
\end{eqnarray*}

\paragraph{Representatives}
Let $\delta,r\in\mathbb{N}$. Let $\mathbf{G},\mathbf{G}'$ be two boundaried (colored) graphs
and let $\bar{R}$ and $\bar{R}'$ be two $r$-tuple of subsets of $V(G)$ and $V(G')$, respectively.
We say that $(\mathbf{G}',\bar{R}')$ is a \emph{$(\delta,r)$-representative} of $(\mathbf{G},\bar{R})$ if 
\begin{itemize}
    \item $(\mathbf{G},\bar{R})$ and $(\mathbf{G}',\bar{R}')$ are compatible; 
    \item $\mathsf{ext\text{-}type}_r(\mathbf{G},\bar{R}) = \mathsf{ext\text{-}type}_r(\mathbf{G}',\bar{R}')$; and
    \item $\mathbf{G}$ and $\mathbf{G}'$ are $\delta$-equivalent.
\end{itemize}

\medskip
The purpose of this section is to prove the following lemma.

\begin{lemma}\label{lemma_gadget}
    Let $\Sigma$ be a colored graph signature.
    There are computable functions $f_k:\mathbb{N}\to\mathbb{N}$ and $f_{\mathsf{rep}}:\mathbb{N}^4\to\mathbb{N}$ and an algorithm that given integers $b,r,q,d,c,\delta\in\mathbb{N}$, an $n$-vertex $b$-boundaried $\Sigma$-colored graph $\mathbf{G}$ and an $r$-tuple~$\bar{R}$ of subsets of $V(G)$, such that
    \begin{itemize}
        \item $G$ admits a tight star decomposition
         $(T,\mathsf{bag})$ such that
    \begin{itemize}
        \item if $s$ is the center of $T$, then $\mathsf{bag}(s)$ is $(q,k)$-unbreakable in $G$ for $k:=f_k(r)$; 
        \item the adhesion of $(T,\mathsf{bag})$ is at most $d$;
        \item $|\mathsf{bag}(\ell)|\le c$ for every leaf $\ell$ of $T$; and
    \end{itemize}
        \item $G$ excludes a graph $H$ as a topological minor,
    \end{itemize}
    outputs, in time $\mathcal{O}_{|H|,|\Sigma|,b,r,q,d,c,\delta}(n^2)$, a $(\delta,r)$-representative $(\mathbf{G}',\bar{R}')$ of $(\mathbf{G},\bar{R})$ of size at most $f_{\mathsf{rep}}(\delta,r,b,|\Sigma|)$, where $\bar{R}'\subseteq \bar{R}$.
\end{lemma}

\medskip 
The proof of this lemma spans the rest of this section. 

\subsection{Subroutines}\label{sec:subroutines}
In this subsection, we prove two main subroutines of our algorithms.

\paragraph{(Extended) folios of boundaried graphs}
The following result shows that, in case the sets~$R_i$, $i\in[r]$ have bounded size, representatives can be found efficiently.

\begin{lemma}\label{lemma_subroutine_folio}
    There is a computable function $f:\mathbb{N}^3\to\mathbb{N}$ such that the following holds. Let $\delta,r,z\in\mathbb{N}$, let $\mathbf{G}$ be a boundaried $\Sigma$-colored graph, and let $\bar{R}=(R_1,\ldots,R_r)$ be an $r$-tuple of subsets of $V(G)$, such that for every $i\in[r]$, $R_i$ has size at most~$z$ and contains the boundary set of $\mathbf{G}$. 
    Then there is a $(\delta,r)$-representative $(\mathbf{G}',\bar{R}')$ of $(\mathbf{G},\bar{R})$ of size at most $f'(\delta,r,z)$, where $\bar{R}'\subseteq \bar{R}$.
    Moreover, 
    such representative can be computed in time $\mathcal{O}_{\delta,r,z}(|\mathbf{G}|^3)$.
\end{lemma}
\begin{proof}
    Let $B$ be the boundary set of $\mathbf{G}$ and let $B^\star:=\bigcup_{i\in[r]} R_i\cup B$.
    Also, let $\mathbf{G}^\star$ be the boundaried graph $(G,B^\star,\rho^\star)$, where $\rho^\star$ is an arbitrary bijection from $B^\star$ to $|B^\star|$.
    We set $b:=|B^\star| \le z\cdot r$, $\delta':=\max\{r,\delta\}$ and $f'(\delta,r,z):=f_{\mathsf{folio}}(\delta',b)$, where $f_{\mathsf{folio}}$ is the function of~\cref{prop:bound-rep-folio}.

    We apply the algorithm of~\cref{prop:bound-rep-folio} for $\delta'$, $b$, and~$\mathcal{F}$. 
    This way, we compute a boundaried graph $\widehat{\mathbf{G}}:=(\widehat{G},\widehat{B},\widehat{\rho})$ of size at most $f_{\mathsf{folio}}(\delta',b)$ such that $\mathbf{G}^\star$ and $\widehat{\mathbf{G}}$ are $\delta'$-equivalent.
    For each $i\in[r]$, let $\widehat{R}_i$ be the vertex subset $\{\widehat{\rho}^{-1}\circ\rho^\star(v)\mid v\in R_i\}$ and $\widehat{B}^-:=\{\widehat{\rho}^{-1}\circ\rho^\star(v)\mid v\in B\}$. We also let $\widehat{\mathbf{G}}':=(\widehat{G},\widehat{B}^-,\widehat{\rho}|_{\widehat{B}^-})$.
    It is easy to observe that $\mathsf{ext\text{-}type}_r(\widehat{\mathbf{G}}',\widehat{R}_1,\ldots,\widehat{R}_r)=\mathsf{ext\text{-}type}_r(\mathbf{G},R_1,\ldots,R_r)$.
    Thus, $(\mathbf{G}',\bar{R}'):=(\widehat{\mathbf{G}}^-,\widehat{R}_1,\ldots,\widehat{R}_r)$ is a $(\delta,r)$-representative of $(\mathbf{G},\bar{R})$ of size at most $f'(\delta,r,z)$ and $\bar{R}'\subseteq \bar{R}$. The claimed running time follows directly from the one given by~\cref{prop:bound-rep-folio}.
\end{proof}

\paragraph{Finding representatives in bounded treewidth graphs}
We also show how to find representatives in bounded treewidth graphs. 
\begin{lemma}\label{lemma_rep_tw}
    Let $\Sigma$ be a colored graph signature.
    There is a function $p:\mathbb{N}^3\to\mathbb{N}$ and an algorithm that given integers $\delta,r,b,t\in\mathbb{N}$, an $n$-vertex $b$-boundaried $\Sigma$-colored graph $\mathbf{G}$ whose underlying graph has treewidth at most $t$ and an $r$-tuple $\bar{R}$ of subsets of $V(G)$, outputs, in time $\mathcal{O}_{|\Sigma|,b,\delta,r,t}(n)$, a $(\delta,r)$-representative  $(\mathbf{G}',\bar{R}')$ of $(\mathbf{G},\bar{R})$ of size at most $p(\delta,r,b,|\Sigma|)$, where $\bar{R}'\subseteq \bar{R}$.
\end{lemma}

The proof of~\Cref{lemma_rep_tw}, which appears in the end of this subsection, uses Courcelle's theorem for a certain MSO-formula. This formula is essentially derived from the formulas defining annotated types. We provide some additional definitions and we show a series of intermediate results.

Let $\Sigma$ be a colored graph signature.
We say that a set $\tau$ of sentences of $\FODP[\Sigma]$ is a \emph{rank-$r$ annotated type} (for $\FODP[\Sigma]$) if there is a $\Sigma$-structure $\mathfrak{A}$ and a tuple $\bar{R}$ of $r$ subsets of the vertex set of $\mathfrak{A}$ such that the rank-$r$ annotated type of $(\mathfrak{A},\bar{R})$ is equal to $\tau$ (up to logical equivalence). The following observation is a reformulation of folklore results on $\mathsf{FO}/\mathsf{MSO}$-types, written for annotated types of $\FODP$;
see~\cite[Lemma 2.1]{Shelah75them},~\cite[Definition 2.2.5]{ebbinghaus1999finite},~\cite[Section 3.4]{libkin2004elements}, and~\cite[Subsection~4.2]{GajarskyGK20diff}.

\begin{observation}\label{obs:formulas-types}
    Let $\Sigma$ be a colored graph signature and let $r\in\mathbb{N}$.
    There is a set of sentences $\mathsf{Type}_{r,\Sigma}$ of $\FODP[\Sigma]$, computable from $r$ and $\Sigma$, such that the following holds.
    For every rank-$r$ annotated type $\tau$, there is a $\psi_\tau\in \mathsf{Type}_{r,\Sigma}$ such that for every $\Sigma$-structure $\mathfrak{A}$ and every $r$-tuple $\bar{R}$ of subsets of the vertex set of $\mathfrak{A}$,
    \begin{displaymath}
        \text{the rank-$r$ annotated type of $(\mathfrak{A},\bar{R})$ is $\tau$ if and only if $(\mathfrak{A},\bar{R})\models \psi_\tau$.}
    \end{displaymath}
\end{observation}

We use the definability of types from the previous observation in order to show the following result. 

\begin{lemma}\label{lemma_mso_formula}
    Let $\Sigma$ be a colored graph signature.
    There is a function $f_{\mathsf{bd\text{-}size}}:\mathbb{N}^3\to\mathbb{N}$ and for every $r,b\in\mathbb{N}$, there is an $\MSO$-formula $\phi_{\Sigma,b,r}(Y_1,\ldots,Y_r)$ such that the following holds.
    Let $r,b\in\mathbb{N}$, let $\mathbf{G}$ be a $b$-boundaried $\Sigma$-colored graph, and let $\bar{R}$ be an $r$-tuple of subsets of $V(G)$.
    For every choice of sets $V_1,\ldots,V_r\subseteq V(G)$, we have that $(\mathbf{G},\bar{R})\models \phi_{\Sigma,b,r}(V_1,\ldots,V_r)$ if and only if
    \begin{itemize}
        \item $\mathsf{ext\text{-}type}_r(\mathbf{G},R_1\cap V_1,\ldots,R_r\cap V_r)=\mathsf{ext\text{-}type}_r(\mathbf{G},\bar{R})$; and 
        \item $|\cup_{i\in[r]} V_i|\le f_{\mathsf{bd\text{-}size}}(|\Sigma|,b,r)$.
    \end{itemize} 
\end{lemma}
\begin{proof}
    We show the lemma by induction on $r$.
    For $r=1$, we set $\Sigma^1$ be the extension of~$\Sigma$, obtained by adding an extra unary relation symbol $U^1$. 
    
    Let $B$ be the boundary of $\mathbf{G}$. For every $I\in\binom{[|B|]}{2}$ and every vertex $v\in R_1$ consider the $\Sigma^1$-structure $\mathbf{G}_v^I$ that is obtained from $\mathbf{G}^I$ by coloring $v$ with color $U^1$. Then note that the graphs in $\{\mathbf{G}_v^I|v\in R_1\}$ can be partitioned into equivalence classes with respect to their rank-$(r-1)$ annotated type. Since the number of equivalence classes is bounded by a constant that depends on $r$ and $\Sigma$, for every $I\in\binom{[|B|]}{2}$ there is a set of vertices $V_1^I\subseteq R_1$, whose size depends only on $r$ and $\Sigma$ and for every set $Y\subseteq R_1$ such that $V_1^I\subseteq Y$ we have that $(\mathbf{G}^I,Y,R_2,\ldots,R_r)$ and $(\mathbf{G}^I,R_1,\ldots,R_r)$ have the same rank-$r$ annotated type. We set $V_1:=\bigcup\{V_1^I: I\in\binom{[|B|]}{2}\}$. Note then that the size of $V_1$ is also depending only on $r,b,$ and $\Sigma$. Also, $(\mathbf{G},V_1,R_1,\ldots,R_r)$ and $(\mathbf{G},R_1,\ldots,R_r)$ have the same extended rank-$r$ annotated type.

    For the inductive step, let $i\in[r]$ and assume that there are vertex subsets $V_1,\ldots,V_{i-1}$ such that 
    \begin{itemize}
        \item for each $j\in[i-1]$, $V_j$ is a subset of $R_j$ whose size depends only on $r,b,$ and $\Sigma$, and
        \item $\mathsf{ext\text{-}type}_r(\mathbf{G},R_1,\ldots,R_r) = \mathsf{ext\text{-}type}_r(\mathbf{G},V_1,\ldots,V_{i-1},R_i,\ldots,R_r).$
    \end{itemize}
    Fix an $I\in\binom{[|B|]}{2}$ and a tuple $\bar{v}=(v_1,\ldots,v_{i-1})$ of $i-1$ vertices such that for each $j\in[i-1]$, $v_j\in V_j$. Note that since the size of each $V_j$ depends only on $r,b,$ and $\Sigma$, the number of different such tuples $\bar{v}$ depends only on $r,b,$ and $\Sigma$. We extend the signature~$\Sigma$ to $\Sigma^i$ by adding $r-i$ unary relation symbols, which will be interpreted as the sets $R_{i+1},\ldots,R_r$.
    Let $\Psi_i$ be the set of all (up to logical equivalence) formulas $\phi(x_1,\ldots,x_i)$ of $\FODP[\Sigma^i]$, where $x_1,\ldots,x_i$ are free variables and $\phi(x_1,\ldots,x_i)$ is of the form
    \[Q_{i+1} x_{i+1}\ldots Q_{r}x_r \psi(x_1,\ldots,x_r)\land\bigwedge_{j=i+1}^{r}x_j\in R_j.\]
    We say that two vertices $u,w\in R_i$ are \emph{$(I,\bar{v})$-equivalent}, which we denote by $u\equiv_{\bar{v}}^{I} w$  if for every formula $\phi(x_1,\ldots,x_i)\in\Psi_i$, we have that
    $(\mathbf{G}^I,\bar{R})\models\phi(v_1,\ldots,v_{i-1},u)$ if and only if $(\mathbf{G}^I,\bar{R})\models\phi(v_1,\ldots,v_{i-1},w)$.
    It is easy to see that $\equiv_{\bar{v}}^{I}$ defines an equivalence relation and the number of equivalence classes depends only on $r,b,$ and $\Sigma$.
    Let $V_{\bar{v}}^{I}\subseteq R_i$ be a vertex subset containing exactly one vertex from each equivalence class of $\equiv_{\bar{v}}^{I}$.
    We set $V_i^{I}:=\bigcup\{V_{\bar{v}}^{I} : \bar{v}\in V_1\times\cdots\times V_{i-1}\}$.
    Note that, by construction, the set $V_i^I$ is a subset of~$R_i$ whose size depends only on $r,b,$ and $\Sigma$ and also 
    \[\mathsf{type}_r(\mathbf{G}^I,V_1,\ldots,V_{i-1},R_i,R_{i+1},\ldots,R_r)=\mathsf{type}_r(\mathbf{G}^I,V_1,\ldots,V_{i-1},V_i^H,R_{i+1},\ldots,R_r).\]
    We finally set $V_i:=\bigcup\{V_i^I: I\in\binom{[|B|]}{2}\}$. Observe then that $V_i$ is a subset of $R_i$ whose size depends only on $r,b,$ and $\Sigma$ and that $(\mathbf{G},R_1,\ldots,R_r)$ and $(\mathbf{G},V_1,\ldots,V_i,R_{i+1},\ldots,R_r)$ have the same extended rank-$r$ annotated type.
    
    Note that the above can be written in an $\MSO$-formula $\phi(V_1,\ldots,V_r)$ using the sentences of~\cref{obs:formulas-types}.    
\end{proof}

As proved by Courcelle~\cite{Courcelle90}, $\MSO$ formulas can be efficiently evaluated on graphs of bounded treewidth.

\begin{theorem}[\cite{Courcelle90}]\label{thm_courcelle}
 There is an algorithm and a computable function $f:\mathbb{N}^{2}\to\mathbb{N}$ that, given a $\MSO$-formula $\varphi(Y_1,\ldots,Y_r)$ and a graph $G$ of treewidth $t\in\mathbb{N}$, runs in time $\mathcal{O}_{\varphi,t}(|G|)$ and finds either sets $V_1,\ldots,V_r\subseteq V(G)$ such that $G\models \varphi(V_1,\ldots,V_r)$ or correctly concludes that such sets do not exist.
\end{theorem}

We conclude this section with the proof of~\Cref{lemma_rep_tw}.

\begin{proof}[Proof of~\cref{lemma_rep_tw}]
    Let $f_{\mathsf{bd\text{-}size}}$ and $\phi_{\Sigma,b,r}(Y_1,\ldots,Y_r)$ be the function and the $\MSO$-formula of~\Cref{lemma_mso_formula}.
    We set $z:=f_{\mathsf{bd\text{-}size}}(|\Sigma|,b,r)$ and $p(\delta,r,b,|\Sigma|):=f'(\delta,r,z)$, where $f'$ is the function of~\cref{lemma_subroutine_folio}. Using~\cref{thm_courcelle}, we can compute the sets $V_1,\ldots,V_r\subseteq V(G)$ in time $\mathcal{O}_{r,\Sigma,t}(n)$. Also, within the same time bound, we can compute the extended rank-$r$ annotated type of $(\mathbf{G},\bar{R})$ and the extended $\delta$-folio of $\mathbf{G}$, since they can be expressed in $\MSO$. Then, because of~\cref{lemma_subroutine_folio}, we know that there is a $(\delta,r)$-representative  $(\mathbf{G}',\bar{R}')$  of $(\mathbf{G}, R_1\cap V_1,\ldots,R_r\cap V_r)$ (and therefore of $(\mathbf{G},\bar{R})$) of size at most $f'(\delta,r,z)$ with $\bar{R}'\subseteq \bar{R}$. Since we know the extended rank-$r$ annotated type and the extended $\delta$-folio, we can compute such a representative in time depending only on $\delta,r,b,$ and $\Sigma$, by brute-force search among all such $(\mathbf{G}',\bar{R}')$ of size at most $f'(\delta,r,z)$ with $\bar{R}'\subseteq \bar{R}$.
\end{proof}

\subsection{Finding small representatives in the presence of large clique minors}
To find small-size representatives, we distinguish two cases.
The first case is when the given graph satisfies the assumptions of ~\autoref{lem:collapse-unbreakable}. In that case, we get that annotated types are $\FO$-definable.
Since $\FO$~model checking is tractable in $K_t$-topological-minor-free classes, in a $K_t$-topological-minor-free (rooted) graph $G$, we can compute its representative in \FPT-time.

\begin{lemma}\label{lem:rep-large-minor}  
    Let $\Sigma$ be a colored graph signature.
    There are functions $f_k:\mathbb{N}\to\mathbb{N}$, $f_t:\mathbb{N}^2\to\mathbb{N},$ and $g_{\mathsf{size}}:\mathbb{N}^3\to\mathbb{N}$, and an algorithm that given integers $b,r,q,d,c,\delta\in\mathbb{N}$, an $n$-vertex $b$-boundaried $\Sigma$-colored graph $\mathbf{G}$, and an $r$-tuple $\bar{R}$ of subsets of $V(G)$ such that
    \begin{itemize}
        \item $G$ admits a tight star decomposition
         $(T,\mathsf{bag})$ such that
    \begin{itemize}
         \item if $s$ is the center of $T$, then
         \begin{itemize}
             \item  $\mathsf{bag}(s)$ is $(q,k)$-unbreakable in $G$ for $k:=f_k(r)$; and 
             \item $\mathsf{bag}(s)$ shatters a $K_t$-minor of $G$ for $t:=f_t(q,r)$;
         \end{itemize}  
        
        \item the adhesion of $(T,\mathsf{bag})$ is at most $d$;
        \item $|\mathsf{bag}(\ell)|\le c$ for every leaf $\ell$ of $T$; and
    \end{itemize}
        \item $G$ excludes a graph $H$ as a topological minor,
    \end{itemize}
    outputs, in time $\mathcal{O}_{|H|,|\Sigma|,b,r,q,d,c,\delta}(n^3)$, a $(\delta,r)$-representative  $(\mathbf{G}',\bar{R}')$ of $(\mathbf{G},\bar{R})$ of size at most $g_{\mathsf{size}}(\delta,r,b,|\Sigma|)$, with $\bar{R}'\subseteq \bar{R}$.
\end{lemma}

\begin{proof}
    Our first goal is to compute sets $R_1',\ldots,R_r'\subseteq V(G)$ such that \begin{itemize}
        \item for each $i\in[r]$, $R_i'$ is a subset of $R_i$ whose size depends only on $b,r,$ and $\Sigma$, and
        \item $\mathsf{ext\text{-}type}_r(\mathbf{G},R_1,\ldots,R_r)=\mathsf{ext\text{-}type}_r(\mathbf{G},R_1',\ldots,R_r')$.
    \end{itemize}
    The rough idea is to use the fact that, under the assumptions of the lemma and because of~\cref{lem:collapse-unbreakable}, annotated types are $\FO$-definable and therefore can be computed using the $\FO$ model checking algorithm of Dvo{\v{r}}{\'a}k, Kr{\'a}l, and Thomas~\cite{dvovrak2010deciding} for topological-minor-free graphs.

    Let $B$ be the boundary of $G$. For every $I\in\binom{[|B|]}{2}$, we denote by $G^I$ the underlying graph of $\mathbf{G}^I$. Note that the conditions concerning the star decomposition in the statement of the lemma also hold for $G^I$, for each $I\in\binom{[|B|]}{2}$, since $(q,k)$-unbreakability is preserved when adding edges. Therefore, for each $I\in\binom{[|B|]}{2}$, the graph $G^I$ satisfies the assumptions of~\cref{lem:collapse-unbreakable}.

    We now show how to compute the sets $R_1',\ldots,R_r'$ inductively.
    Let $i\in[r]$ and assume that we have already computed vertex subsets $R_1',\ldots,R_{i-1}'$ such that 
    \begin{itemize}
        \item for each $j\in[i-1]$, $R_j'$ is a subset of $R_j$ whose size depends only on $b,r,$ and $\Sigma$, and
        \item $\mathsf{ext\text{-}type}_r(\mathbf{G},R_1,\ldots,R_r)=\mathsf{ext\text{-}type}_r(G,R_1',\ldots,R_{i-1}',R_i,\ldots,R_r)$.
    \end{itemize}
    We fix an $I\in\binom{[|B|]}{2}$ and a tuple $\bar{v}=(v_1,\ldots,v_{i-1})$ of $i-1$ vertices of $G$ such that for each $j\in[i-1]$, $v_j\in R_j'$. 
    Note that since the size of each $R_j'$ depends only on $b,r,$ and~$\Sigma$, the number of different such tuples $\bar{v}$ depends only on $b,r,$ and $\Sigma$. 
    We extend the signature~$\Sigma$ to $\Sigma^i$ by adding $r-i$ unary relation symbols, which will be interpreted as the sets $R_{i+1},\ldots,R_r$.
    Let $\Psi_i$ be the set of all (up to logical equivalence) formulas $\phi(x_1,\ldots,x_i)$ of $\FODP[\Sigma^i]$, where $x_1,\ldots,x_i$ are free variables and $\phi(x_1,\ldots,x_i)$ is of the form
    \[Q_{i+1} x_{i+1}\ldots Q_{r}x_r \psi(x_1,\ldots,x_r)\land\bigwedge_{j=i+1}^{r}x_j\in R_j.\]
    We say that two vertices $u,w\in R_i$ are \emph{$(I,\bar{v})$-equivalent}, which we denote by $u\equiv_{\bar{v}}^{I} w$  if for every formula $\phi(x_1,\ldots,x_i)\in\Psi_i$, we have that
    $(\mathbf{G}^I,\bar{R})\models\phi(v_1,\ldots,v_{i-1},u)$ if and only if $(\mathbf{G}^I,\bar{R})\models\phi(v_1,\ldots,v_{i-1},w)$.
    It is easy to see that $\equiv_{\bar{v}}^{I}$ defines an equivalence relation and the number of equivalence classes depends only on $b,r,$ and $\Sigma$.
    Let $Q_{\bar{v}}^{I}\subseteq R_i$ be a vertex subset containing exactly one vertex from each equivalence class of $\equiv_{\bar{v}}^{I}$.
    We claim that the set $Q_{\bar{v}}^{I}$ can be computed in time $\mathcal{O}_{|H|,|\Sigma|,b,r,q,d,c}(n^2)$. 
    To see this, note that since~$G^I$ satisfies the assumptions of~\cref{lem:collapse-unbreakable}, 
    every formula $\phi(x_1,\ldots,x_i)\in\Phi_i$ can be rewritten to an $\FO$-formula $\hat{\phi}(x_1,\ldots,x_i)$ by replacing in $\phi$ each (sub)formula of the form $\DP_r[(x_1,y_1),\ldots, (x_r,y_r)]$ with the $\FO$-formula $\phi_{\mathsf{dp}}(x_1,y_1,\ldots,x_r,y_r)$ from~\cref{lem:collapse-unbreakable}, whose size depends on $r,q,d,$ and $c$.
    Then, by applying the model checking algorithm of Dvo{\v{r}}{\'a}k, Kr{\'a}l, and Thomas~\cite{dvovrak2010deciding}, we can check for each $v\in R_i$, which are the formulas $\phi(x_1,\ldots,x_i)\in\Psi_i$ such that $\mathbf{G}^I\models\hat{\phi}(v_1,\ldots,v_{i-1},v)$. This way, we can partition $R_i$ to equivalence classes with respect to $\equiv_{\bar{v}}^{I}$ in time $\mathcal{O}_{|H|,|\Sigma|,b,r,q,d,c}(n^2)$ and pick one vertex from each equivalence class to obtain $Q_{\bar{v}}^{I}$.
    We set $Q_i^{I}:=\bigcup\{Q_{\bar{v}}^{I} : \bar{v}\in R_1'\times\cdots\times R_{i-1}'\}$.
    Note that, by construction, the set $Q_i^I$ is a subset of $R_i$ whose size depends only on $b,r,$ and $\Sigma$ and also 
    \[\mathsf{type}_r(\mathbf{G}^I,R_1',\ldots,R_{i-1}',R_i,R_{i+1},\ldots,R_r)=\mathsf{type}_r(\mathbf{G}^I,R_1',\ldots,R_{i-1}',Q_i^H,R_{i+1},\ldots,R_r).\]
    We finally set $R_i':=\bigcup\{Q_i^I: I\in\binom{[|B|]}{2}\}$. Observe then that $R_i'$ is a subset of $R_i$ whose size depends only on $b,r,$ and $\Sigma$ and that $(\mathbf{G},R_1,\ldots,R_r)$ and $(\mathbf{G},R_1',\ldots,R_i',R_{i+1},\ldots,R_r)$ have the same extended rank-$r$ annotated type.
    
    By applying inductively the above argument for each $i\in[r]$, we compute the claimed sets $R_1',\ldots,R_r'$ in time $\mathcal{O}_{|H|,|\Sigma|,b,r,q,d,c}(n^2)$.
    
    Finally, using the algorithm of~\cref{lemma_subroutine_folio}, in time $\mathcal{O}_{|\Sigma|,b,r,\delta}(n^3)$, we compute a $(\delta,r)$-representative $(\widehat{\mathbf{G}},\widehat{R}_1,\ldots,\widehat{R}_r)$ of $(\mathbf{G},R_1',\ldots,R_r')$ (that is also a $(\delta,r)$-representative of $(\mathbf{G},\bar{R})$) whose size depends only on $\delta,r,b,$ and~$\Sigma$; also, for each $i\in[r]$, $\widehat{R}_i\subseteq R_i$.
\end{proof}

\subsection{Excluding large clique minors}

We now deal with the case where the given graph does \textsl{not} contain large clique minors. We show the following result.

\begin{lemma}\label{lem:rep-excluding-minor}
    There is a computable function $g_\mathsf{size}^\prime:\mathbb{N}^4\to\mathbb{N}$ and an algorithm with the following specifications. Let $b,d,c,t,r,\delta\in\mathbb{N}$, let  $\mathbf{G}$ be an $n$-vertex $b$-boundaried $\Sigma$-colored graph, let $\bar{R}$ be an $r$-tuple of subsets of $V(G)$, and let  $(T,\mathsf{bag})$  be a star decomposition of $G$ of adhesion at most $d$ such that for every leaf $\ell$ of $T$, we have $|\mathsf{bag}(\ell)|\le c$.
    Then in time $\mathcal{O}_{|\Sigma|,b,d,c,t,r,\delta}(n^3)$, we can either:
    \begin{itemize}
        \item report that $\mathsf{bag}(s)$ shatters a $K_t$-minor of $G$; or
        \item compute a $(\delta,r)$-representative  $(\mathbf{G}',\bar{R}')$ of $(\mathbf{G},\bar{R})$ of size at most $g_\mathsf{size}^\prime(\delta,r,b,|\Sigma|)$, with $\bar{R}'\subseteq \bar{R}$.
    \end{itemize}
\end{lemma}

The proof of~\Cref{lem:rep-excluding-minor} is based on the irrelevant vertex technique. We use as a black box the subroutine of the algorithm of~\cite{golovach2022model}. In order to state this result, we need to define walls and related notions.

\paragraph{Walls}
Let $h$ be an even integer with $h\ge 2$. The \emph{elementary wall of height $h$} is the graph~$\overline{W}_h$ with vertex set $\{0,\ldots,2h+1\}\times\{0,\ldots,h\}\setminus\{(0,0),(2h+1,h)\}$ and an edge between any vertices $(i,j)$ and $(i',j')$ if either
\begin{itemize}
    \item $|i-i'|=1$ and $j=j'$, or
    \item $i=i'$, $|j-j'|=1$, and $i$ and $\max\{j,j'\}$ have the same parity.
\end{itemize} 
Given an odd $i\ge 1$, the \emph{$i$th vertical path} of $\overline{W}_h$ is the path of $\overline{W}_h$ with vertices $(i,j-1),(i,j),(i-1,j),(i-1,j+1)$ for every odd $j\ge 1$. Given an $i\ge 1$, the \emph{$i$th horizontal path} of $\overline{W}_h$ is the path of $\overline{W}_h$ with all vertices of $\overline{W}_h$ of the form $(j,i-1)$ where $0\le j\le 2h+1$.

A \emph{wall of height $h$} (for even $h\ge 2$) is a subdivision $W$ of the elementary wall $\overline{W}_h$ of height $h$. The \emph{vertical/horizontal paths} of $W$ are defined as the paths of $W$ corresponding to the subdivisions of the vertical/horizontal paths of $\overline{W}_h$. Given a graph $G$, a \emph{wall in $G$} is a subgraph of $G$ which is isomorphic to a wall. 

The following result is the classical \emph{Grid Theorem} of Robertson and Seymour~\cite{robertson1986graph}, stated for walls instead of grids.

\begin{proposition}\label{walltheorem}
    There is a computable function $f:\mathbb{N}\to\mathbb{N}$ such that for every even $h\ge 2$, every graph of treewidth at least $f(h)$ contains a wall of height $h$ as a subgraph.
\end{proposition}

In the rest, we need the notion of a clique minor being ``gathered around'' a wall.
Let~$W$ be a wall of height $h$ in a graph $G$. Also, let $P_1,\ldots,P_h$ be the vertical paths of~$W$ and $Q_1,\ldots,Q_h$ be the horizontal paths of $W$.
We say that a minor model of $K_t$ (for some $t\in\mathbb{N}$) in $G$ is \emph{grasped} by $W$ if for every branch set $B_k$ of the model, there exist distinct indices $i_1,\ldots,i_t\in\{1,\ldots,h\}$ and distinct indices $j_1,\ldots,j_t\in\{1,\ldots,h\}$ such that $V(P_{i_\ell})\cap V(Q_{j_\ell})\subseteq B_k$ for every $\ell\in\{1,\ldots,t\}$.

To prove~\autoref{lem:rep-excluding-minor}, we need the following result from~\cite{golovach2022model}. Given a boundaried graph $\mathbf{G}=(G,B,\rho)$ and a vertex subset $X\subseteq V(G)$ with $B(\mathbf{G})\subseteq X$, we use $\mathbf{G}[X]$ to denote the boundaried graph $(G[X],B,\rho)$.

\begin{proposition}[\cite{golovach2022model}]\label{prop:mc-minor}
There is a computable function $h\colon\mathbb{N}^2\to\mathbb{N}$
and an algorithm that given
$r,t\in\mathbb{N}$,
a $b$-boundaried $\Sigma$-colored graph $\mathbf{G}$, a tuple $(R_1,\ldots,R_r)$ of $r$ subsets of $V(G)$, and a wall $W$ of $G$ of height $h(t,r)$, 
outputs, in time ${\cal O}_{|\Sigma|,b,t,r}(n)$, one of the following:
\begin{itemize}
\item a report that $G$ has a minor model of $K_{t}$ that is grasped by $W$,
or
\item sets $V,R_1',\ldots,R_r'\subseteq V(G)$ with $B(\mathbf{G})\subseteq V\neq V(G)$ and $R_i'\subseteq R_i\cap V$ for each $i\in[r]$,
such that $\mathsf{ext\text{-}type}_r(\mathbf{G},R_1,\ldots,R_r)=\mathsf{ext\text{-}type}_r(\mathbf{G}[V],R_1',\ldots,R_r')$.
\end{itemize}
\end{proposition}

We stress that~\cref{prop:mc-minor} does not explicitly appear in~\cite{golovach2022model}. The version of this result appearing in~\cite{golovach2022model} uses as a blackbox the algorithmic version of the \emph{Flat Wall theorem} from~\cite[Theorem 8]{SauST24amor}, which in turn is based on the Flat wall theorem of Kawarabayashi, Thomas, and Wollan~\cite[Theorem 7.7]{KawarabayashiTW18anew}. The original statement of the latter reports a minor model of $K_{t}$ that is grasped by $W$ but in \cite[Theorem 8]{SauST24amor}, this part is weakened to just report a $K_t$ minor. By using the original statement of~\cite[Theorem 7.7]{KawarabayashiTW18anew}, one can very easily modify the algorithm of~\cite[Theorem 8]{SauST24amor} to report that $G$ has a minor model of $K_{t}$ that is grasped by $W$ (instead of just reporting that $G$ has $K_t$ as a minor). This change directly propagates to the results of~\cite{golovach2022model}. To reprove this, one should apply these local changes in this chain of results. Since this is a trivial task, that would however require long series of definitions, we omit the proof of~\cref{prop:mc-minor} here.  

Using~\cref{prop:mc-minor} and~\cref{lemma_rep_tw}, we can show~\cref{lem:rep-excluding-minor}.

\begin{proof}[Proof of~\autoref{lem:rep-excluding-minor}]
    First, we apply the algorithm of~\cref{prop:bound-rep-folio} to compute a graph~$\mathbf{H}$ that is $\delta$-equivalent to $\mathbf{G}$. This is done in time $\mathcal{O}_{\delta,b}(n^3)$.

    We next show how to bound the treewidth of $G$.
    Let $h:=\max\{c+1,h(t,r),d^2+1\}$, where $h$ is the function of~\cref{prop:mc-minor}.
    Suppose that the treewidth of $G$ is at least $g(h)$. Then, because of~\cref{walltheorem}, $G$ contains a wall $\widehat{W}$ of height $h$ as a subgraph.
    Also note that at most $\max\{c,d^2\}$ of the non-subdivision vertices of $\widehat{W}$ can be in $V(G)\setminus\mathsf{bag}$.
    By applying the algorithm of~\Cref{prop:mc-minor}, in time ${\cal O}_{|\Sigma|,b,r,t,c,d}(n)$, we get one of the following:
    \begin{itemize}
    \item a report that $G$ has a minor model of $K_{t}$ that is grasped by $W$,
        or
    \item sets $V,R_1',\ldots,R_r'\subseteq V(G)$ with $B(\mathbf{G})\subseteq V\neq V(G)$ and $R_i'\subseteq R_i\cap V$ for each $i\in[r]$, such that $\mathsf{ext\text{-}type}_r(\mathbf{G},R_1,\ldots,R_r)=\mathsf{ext\text{-}type}_r(\mathbf{G}[V],R_1',\ldots,R_r')$.
    \end{itemize}
    In the first case, note that the fact that at most $\max\{c,d^2\}$ of the non-subdivision vertices of $\widehat{W}$ can be in $V(G)\setminus\mathsf{bag}$ implies that every branch set of the model of $K_t$ interesects $\mathsf{bag}(s)$, and therefore we can safely report that $\mathsf{bag}(s)$ shatters a $K_t$-minor of~$G$. 
    In the latter case, we recursively apply the above procedure for $(\mathbf{G}[V],R_1',\ldots,R_r')$. This way, after $\mathcal{O}(n)$ iterations, we get a set $V^\star\subseteq V(G)$ that contains $B(\mathbf{G})$ and sets $R_1^\star,\ldots,R_r^\star\subseteq V^\star$ such that $(\mathbf{G},R_1,\ldots,R_r)$ and $(\mathbf{G}[V^\star],R_1^\star,\ldots,R_r^\star)$  have the same extended rank-$r$ annotated type and the treewidth of $G[V^\star]$ is less than $g(h)$. The total running time of this procedure is ${\cal O}_{|\Sigma|,b,t,r,c,d}(n^2)$.

    We next compute a $(\delta,r)$-representative $(\widetilde{\mathbf{G}},\bar{R}')$ of $(\mathbf{G}[V^\star],R_1^\star,\ldots,R_r^\star)$, by applying the algorithm of~\Cref{lemma_rep_tw}, whose running time is ${\cal O}_{|\Sigma|,b,r,g(h)}(n)$. We also get that $\bar{R}'\subseteq \bar{R}$. Then, consider the boundried graph $\mathbf{G}'$ that is obtained from $\widetilde{\mathbf{G}}$ after enhancing its underlying graph $\widetilde{G}$ by taking ts disjoint union with the underlying graph of $\mathbf{H}$ (the boundaried graph that is $\delta$-equivalent to $\mathbf{G}$)  and identifying the vertices on their common boundary.
    Note that $(\mathbf{G}',\bar{R}')$ is a $(\delta,r)$-representative of $(\mathbf{G},\bar{R})$ whose size depends only on $\delta,r,b,$ and $\Sigma$.
\end{proof}

\subsection{Computing representatives - Proof of~\texorpdfstring{\Cref{lemma_gadget}}{Lemma 4.1}}

We set $k:=f_k(r)$ and $t:=f_t(q,r)$, where $f_k,f_t$ are the corresponding functions from~\cref{lem:rep-large-minor}.
We set
\begin{eqnarray}
   & c_1 := g_{\mathsf{size}}(\delta,r,b,|\Sigma|),  &\text{where $g_{\mathsf{size}}$ is the function of~\cref{lem:rep-large-minor}, and}\notag \\
    & c_2:=g_{\mathsf{size}}^\prime(\delta,r,b,|\Sigma|), & \text{where $g_{\mathsf{size}}^\prime$ is the function of~\cref{lem:rep-excluding-minor}}.\notag
\end{eqnarray}
We then set $f_{\mathsf{rep}}(\delta,r,b,|\Sigma|):=\max\{c_1,c_2\}$.

We apply the algorithm of~\cref{lem:rep-excluding-minor} for the given values of $b,d,c,t,r$, the given $b$-boundaried $\Sigma$-colored graph $\mathbf{G}$ and the given $r$-tuple $\bar{R}$ of subsets of $V(G)$. 
This runs in time $\mathcal{O}_{|\Sigma|,b,d,c,t,r}(n^3)$.
If it outputs an $r$-representative of $(\mathbf{G},\bar{R})$ of size at most $c_2$, then we return this. If it reports that $\mathsf{bag}(s)$ shatters a $K_t$-minor of $G$, then we apply the algorithm of~\cref{lem:rep-large-minor}. 
Since $\mathsf{bag}(s)$ is $(q,k)$-unbreakable for $k=f_k(r)$ and shatters a $K_t$-minor of~$G$ for $t=f_t(q,r)$ and $G$ excludes a graph $H$ as a topological minor (by assumption), this algorithm returns, in time $\mathcal{O}_{|H|,|\Sigma|,b,r,q,d,c}(n^3)$, an $r$-representative of $(\mathbf{G},\bar{R})$ of size at most $c_1$. 
The overall running time is $\mathcal{O}_{|H|,|\Sigma|,b,r,q,d,c}(n^3)$.

\section{Compositionality of types}\label{sec:compositionality}

We now prove the following ``Feferman-Vaught style'' composition lemma. It intuitively says for every separation $(X,Y)$ of a graph $G$,
one can safely replace~$G[X]$ rooted at $X\cap Y$ with some representative of~$G[X]$ without affecting the extended annotated type of the whole graph.
This is one of the key arguments for the correctness of the dynamic programming algorithm described in~\Cref{subsec:main-proof}.

Before stating the lemma, we define the gluing of boundaried graphs of the form~$(\mathbf{G},\bar{R})$.
Let $\mathbf{G}_{1}=(G_1,B_1,\rho_1)$ and $\mathbf{G}_{2}=(G_2,B_2,\rho_2)$ be two compatible boundaried graphs.
We define $(\mathbf{G}_{1},\bar{R}_1)\oplus(\mathbf{G}_{2},\bar{R}_2)$ as the pair $(G,\bar{R})$, where
\begin{itemize}
    \item $G$ is the graph obtained if we take the disjoint union of $G_{1}$ and $G_{2}$ and, for every $i\in[|B_{1}|],$ we identify the vertices $\rho_{1}^{-1}(i)$ and $\rho_{2}^{-1}(i)$, agreeing that after each such identification the vertices of the boundary of $\mathbf{G}_{1}$ prevail, and
    \item $\bar{R}$ is the $r$-tuple $(R_1,\ldots,R_r)$, where for each $i\in[r]$, $R_i$ is obtained by the disjoint union of $R_i^1$ and $R_i^2$ and, for every $i\in[|B_{1}|],$ if $\rho_{1}^{-1}(i)\in R_i^1$, we identify the vertices $\rho_{1}^{-1}(i)$ and $\rho_{2}^{-1}(i)$, agreeing that after each such identification the vertices of $R_i^1\cap B_1$ prevail. 
\end{itemize}

\begin{lemma}\label{lemma_compositionality}
    Let $r\in\mathbb{N}$, let $\mathbf{G}_1,\mathbf{G}_2$ be two boundaried graphs and let $\bar{R}_1,\bar{R}_2$ be two $r$-tuple of subsets of $V(G_1)$ and of $V(G_2)$, respectively. Suppose that
    \begin{itemize}
        \item $(\mathbf{G}_1,\bar{R}_1)$ and $(\mathbf{G}_2,\bar{R}_2)$ are compatible; and
        \item $\mathsf{ext\text{-}type}_r(\mathbf{G}_1,\bar{R}_1) =\mathsf{ext\text{-}type}_r(\mathbf{G}_2,\bar{R}_2).$
    \end{itemize} 
    Then for every boundaried graph $\mathbf{C}$ and every $r$-tuple $\bar{Q}$ of subsets of $V(C)$ such that $(\mathbf{C},\bar{Q})$ is compatible with $(\mathbf{G}_1,\bar{R}_1)$ (and $(\mathbf{G}_2,\bar{R}_2)$), it holds that 
    \[\mathsf{type}_r\left((\mathbf{C},\bar{Q})\oplus(\mathbf{G}_1,\bar{R}_1)\right)=\mathsf{type}_r\left((\mathbf{C},\bar{Q})\oplus(\mathbf{G}_2,\bar{R}_2)\right).\]
\end{lemma}

\begin{proof}
We set $(R_1^1,\ldots,R_r^1):=\bar{R}_1$ and $(R_1^2,\ldots,R_r^2):=\bar{R}_2$.
Since $(\mathbf{G}_1,\bar{R}_1)$ and $(\mathbf{G}_2,\bar{R}_2)$ are compatible, $|B(\mathbf{G}_1)|=|B(\mathbf{G}_2)|$. We set $b:=|B(\mathbf{G}_1)|$.
Our goal is to show the equivalent statement:
for every sentence $\varphi$ of $\FODP[\Sigma]$ in prenex normal form that has quantifier rank at most $r$,
\[(\mathbf{C},\bar{Q})\oplus(\mathbf{G}_1,\bar{R}_1)\models \phi \qquad \text{if and only if}\qquad (\mathbf{C},\bar{Q})\oplus(\mathbf{G}_2,\bar{R}_2)\models\phi,\]
when the $i$th variable of $\varphi$ is interpreted as an element of $R_i^1$ (resp. $R_i^2$).
We denote by~$G$ and $G'$ the (uncolored) underlying graph of $(\mathbf{C},\bar{Q})\oplus(\mathbf{G}_1,\bar{R}_1)$ and $(\mathbf{C},\bar{Q})\oplus(\mathbf{G}_2,\bar{R}_2)$, respectively.

This is shown by induction on the structure of $\varphi$. In fact, the proof is standard for all atomic predicates except the disjoint-paths predicate, for boolean connectives, as well as first-order quantifiers, see e.g., \cite[Lemma~2.3]{Grohe2007LogicGA}.
In the rest of the proof, we focus on the disjoint-paths predicate. In particular, we show the following claim. To make the presentation more accesible, we assume that $B(\mathbf{G}_1)=B(\mathbf{G}_2)=B(\mathbf{C})\subseteq V(G_1)$ and we denote this set by $B$. Also, given two tuples $\bar{x},\bar{y}$, we denote by $\bar{x}\cup\bar{y}$ the set containing all elements appearing either in $\bar{x}$ or in $\bar{y}$.

\begin{claim}
  We consider the following tuples of vertices:
  \begin{itemize}
    \item $\bar{x}$ is a tuple of $i_1$ vertices in $V(C)\setminus B$;
    \item $\bar{z}$ is a tuple of $i_2$ vertices in $B$; and
    \item $\bar{y}$ is a tuple of $i_3$ vertices in $V(G_1)\setminus B$,
  \end{itemize}
  such that $i_1+i_2+i_3=r$.
  Then, there is a tuple $\bar{y}':=(y_1',\ldots,y_{i_3}')$ of $i_3$ vertices in $V(G_2)\setminus B$ such that
  \begin{itemize}
    \item for every $i\in[i_3]$ and every $j\in [r]$, $y_i\in R_j^1$ if and only if $y_i'\in R_j^2$; and   
    \item for every collection $(s_1,t_1),\ldots,(s_r,t_r)$ of pairs of vertices from $\bar{x}\cup\bar{z}\cup\bar{y}$, we have 
  \[G\models \DP_r[(s_1,t_1),\ldots, (s_r,t_r)]\qquad \text{if and only if}\qquad G'\models\DP_r[(s_1',t_1'),\ldots, (s_r',t_r')],\] 
  where $(s_1',t_1'),\ldots, (s_r',t_r')$ is obtained from $(s_1,t_1),\ldots,(s_r,t_r)$ by repacing every occurance of $y_i$ with the corresponding $y_i'$, i.e., $s_i'=y_j'$ (resp. $t_i'=y_j'$) if $s_i=y_j$ (resp. $t_i=y_j$) for some $j\in[i_3]$ and otherwise $s_i'=s_i$ (resp. $t_i'=t_i$).
  \end{itemize} 
\end{claim}
\begin{claimproof}
Suppose that $G\models\DP_r[(s_1,t_1),\ldots,(s_r,t_r)]$.
This implies the existence of pairwise internally vertex-disjoint paths $P_1,\ldots,P_r$ in $G$,
where for every $i\in[r]$,
$P_i$ is an $(s_i,t_i)$-path.
The following arguments are inspired by the proof of~\cite[Lemma 2.4]{grohe2011finding}.
We define a graph $F^\star$ on $B\cap \bigcup_{i\in[r]}V(P_i)$ such that two vertices $a,b$ of $F^\star$ are adjacent if there is some $i\in[r]$ and a subpath of $P_i$ with endpoints $a$ and $b$ and every internal vertex in $V(C)\setminus B$.
For every edge $ab\in E(F^\star)$, we denote by $P_{ab}$ this subpath.
Let also \[I:=\big\{\{\rho(a),\rho(b)\}: ab\in E(F^\star)\big\}\]
and let $G_1^I$ and $G_2^I$ be the underlying graphs of $\mathbf{G}_1^I$ and $\mathbf{G}_2^I$, respectively.

For every path~$P$ in $G$ that has endpoints in~$G_1$, we write $\mathsf{compressed}(P)$ to refer
the path of $G_1^I$ obtained by replacing subpaths of $P$ whose internal vertices are in $V(C)\setminus B$ by the appropriate edges of~$F^\star$.
Similarly, if $Q$ is a path of~$G_1^I$, we write $\mathsf{expanded}(Q)$ to refer to the path of $G$ obtained by replacing each edge $ab$ of~$F^\star$ by the corresponding path $P_{ab}$.
For every $i\in[r]$,
we define a path $\hat{P}_i$ as follows:
\begin{itemize}
\item if $s_i,t_i\in V(G_1)$, then $\hat{P}_i=\mathsf{compressed}(P_i)$,
\item if $s_i\in V(G_1)$ and $t_i\in V(C)\setminus B$ then $\hat{P}_i$ is the $(s_i,w_i)$-subpath of $P_i$, where $w_i$ is the last vertex of $P_i$ (traversing from $s_i$ to $t_i$) that belongs to $V(G_1)$; we define  $\hat{P}_i$ analogously when $t_i\in V(G_1)$ and $s_i\in V(C)\setminus B$,
\item if $s_i,t_i\in V(C)\setminus B$, and $V(P_i)\subseteq V(C)\setminus B$, then $\hat{P}_i$ is defined as the empty graph,
\item if $s_i,t_i\in V(C)\setminus B$ and $P_i$ intersects $V(G_1)$ at a single vertex $w$ (in $B$), then $\hat{P}_i$ is the trivial $(w,w)$-path,
\item if $s_i,t_i\in V(C)\setminus B$ and $P_i$ intersects $V(G_1)$ at least twice, then if $w$ and $u$ are the first and the last vertex of $P_i$ in $V(G_1)$, then $\hat{P}_i$ is the $(u,w)$-subpath of $\mathsf{compressed}(P_i)$.
\end{itemize}

Observe that the paths $\hat{P}_1,\ldots,\hat{P}_r$ are pairwise internally vertex-disjoint paths in~$G_1^I$.
Also, note that the fact that $\mathsf{ext\text{-}type}_r(\mathbf{G}_1,\bar{R}_1) =\mathsf{ext\text{-}type}_r(\mathbf{G}_2,\bar{R}_2)$ implies that
\[\mathsf{type}_r(\mathbf{G}_1^I,\bar{R}_1) = \mathsf{type}_r(\mathbf{G}_2^I,\bar{R}_1).\]
Therefore, the existence of the paths $\hat{P}_1,\ldots,\hat{P}_r$ in $G_1^I$ implies 
the existence of a tuple of vertices $\bar{y}'=(y_1',\ldots,y_{i_3}')$ and a collection of paths $\hat{Q}_1,\ldots,\hat{Q}_r$, with the following properties:
\begin{itemize}
    \item $\hat{Q}_1,\ldots,\hat{Q}_r$ are pairwise internally vertex-disjoint paths in $G_2^I$.
    \item if $\hat{P}_i$ has endpoints $u,w$ then $\hat{Q}_i$ has endpoints~$u',w'$, where $u'=y_j'$ if $u=y_j$ (resp. $w'=y_j'$ if $w=y_j$) for some $j\in[i_3]$,
and $u'=u$ (resp. $w'=w$) otherwise.
\end{itemize} 
Note that, if we set $Q_i =\mathsf{expanded}(\hat{Q}_i)$ for each $i\in[r]$,
then each $Q_i$ is an $(s_i',t_i')$-path and $Q_1,\ldots, Q_r$ are pairwise internally vertex-disjoint paths of $G'$.
Therefore $G'\models\DP_r[(s_1',t_1'),\ldots,(s_r',t_r')]$.

For the reverse implication, i.e., that $G'\models\DP_r[(s_1',t_1'),\ldots,(s_r',t_r')]$ implies $G\models\DP_r[(s_1,t_1),\ldots,(s_r,t_r)]$, the proof is symmetric.
\end{claimproof}
\end{proof}

In order to show the correctness of our algorithm, we also need to show that when replacing cones with small representatives, we still exclude a topological minor after the replacement. This is materialized in~\cref{lem:excluded-top-minor},
is a direct corollary of~\Cref{lemma_compositionality} and the fact that graphs with the same extended rank-$2h$ annotated type are $h$-equivalent. The latter can be easily derived from the fact that the existence of a graph $H$ with $\|H\|\le h$ as a topological minor can be expressed in an $\FODP$ formula of quantifier rank at most $2h$.

\begin{observation}\label{obs_types_to_folios}
    Let $h\in\mathbb{N}$, let $\mathbf{G},\mathbf{G}'\in\mathcal{B}$, and let $\bar{R},\bar{R}'$ be two $h^2$-tuples subsets of $V(G)$ and $V(G')$, respectively. If $(\mathbf{G},\bar{R})$ and $(\mathbf{G}',\bar{R}')$ have the same extended rank-$2h$ annotated type, then $\mathbf{G}$ and $\mathbf{G}'$ are $h$-equivalent.
\end{observation}

Given a graph $G$, a tree decomposition $\mathcal{T}=(T,\mathsf{bag})$, and a node $t\in V(T)$,
$\mathbf{G}_t$ denotes the boundaried graph $(G[\cone(t)],\adh(t),\rho)$, for some $\rho\colon \adh(t)\to|\adh(t)|$. We also use~$G_t$ to denote the graph $G[\mathsf{cone}(t)]$.

\begin{corollary}\label{lem:excluded-top-minor}
    Let $c,h\in\mathbb{N}$.
    Let $G$ be a graph and let $\mathcal{T}=(T,\mathsf{bag})$ be a tree decomposition of~$G$.
    Let $t$ be a non-root node in $V(T)$.
    Let~$\mathbf{G}'$ be a boundaried graph of order at most $c$ such that $\mathbf{G}_t$ and $\mathbf{G}'$ are $h$-equivalent and let~$\hat{G}$ be the graph obtained from $G$ by replacing $G_t$ with the underlying graph of $\mathbf{G}'$.
    Then, if $G$ excludes a graph $H$ of size at most $h$ as a topological minor, the same holds for $\hat{G}$.
\end{corollary}

\section{The model checking algorithm}
\label{subsec:main-proof}
In this section, we present our model checking algorithm and we prove its correctness. To show its correctness, we need to show that unbreakability and tightness are preserved when we do a replacement. This is shown in the first subsection. Then, we proceed to the proof of~\cref{main_theorem}.

\subsection{Preserving unbreakability}

Recall that, given a graph $G$, a tree decomposition $\mathcal{T}=(T,\mathsf{bag})$, and a node $t\in V(T)$,
$\mathbf{G}_t$ denotes the boundaried graph $(G[\cone(t)],\adh(t),\rho)$, for some $\rho\colon \adh(t)\to|\adh(t)|$. We also use $G_t$ to denote the graph $G[\mathsf{cone}(t)]$.
We next show that in a tree decomposition of adhesion~$\ell$, if we replace $G_t$ with a graph with the same $\ell$-folio, we have that strong unbreakability is preserved in the obtained tree decomposition.

\begin{lemma}\label{obs_maintain_regularity}
    Let $\ell\in\mathbb{N}$.
    Let $G$ be a graph and let $\mathcal{T}=(T,\mathsf{bag})$ be a tree decomposition of~$G$ of adhesion~$\ell$.
    Let~$t$ be a node in $V(T)$.
    Let~$\mathbf{G}'$ be a boundaried graph such that~$\mathbf{G}_t$ and~$\mathbf{G}'$ have the same $\ell$-folio.
    Also, let $\hat{G}$ be the graph obtained from $G$ by replacing $G_t$ with $G'$.
    Let $\hat{\mathcal{T}}$ be the tree decomposition of $\hat{G}$ obtained by~$\mathcal{T}$ after removing all (strict) descendants of $t$ and setting $\mathsf{bag}(t)=V(G')$.
    Then, for every node $x\neq t$ of $\hat{\mathcal{T}}$, it holds that if $x$ has the strong $(q,k)$-unbreakability property in $\mathcal{T}$, then it also has it in $\hat{\mathcal{T}}$.
\end{lemma}
\begin{proof}
    Let $x$ be a node of $\hat{\mathcal{T}}$ that is not $t$ and let $G_x$ (resp.~$\hat{G}_x$) be the graph induced by the cone of $x$ in $\mathcal{T}$ (resp.~$\hat{\mathcal{T}}$).
    Note that, by definition of $\hat{G}$ and $\hat{\mathcal{T}}$, $\mathsf{bag}(x)$ is the same in both $\mathcal{T}$ and $\hat{\mathcal{T}}$.
    Our goal is to show that if $\mathsf{bag}(x)$ is $(q,k)$-unbreakable in $G_x$, then it is also $(q,k)$-unbreakable in $\hat{G}_x$.

    Assume towards a contradiction that there is a separation $(A,B)$ of $\hat{G}_x$ of order at most $k$ such that $|A\cap \mathsf{bag}(x)|>q$ and $|B\cap \mathsf{bag}(x)|>q$.
    Then, consider a minimum-order separation $(A^\star,B^\star)$ of $G_x$ with $A\cap \mathsf{bag}(x)\subseteq A^\star$ and $B\cap\mathsf{bag}(x)\subseteq B^\star$.
    
    We claim that $|A^\star\cap B^\star|\le k$. To see this, observe that $A^\star\cap B^\star$ intersects at most $\ell$ vertices of $G[\mathsf{cone}(t)]$, since otherwise we could replace the part of $A^\star\cap B^\star$ in $\mathsf{cone}(t)$ with $\mathsf{adh}(t)$ and still get a separation with the claimed properties and of smaller order. Also, recall that in $\hat{G}_x$, the graph $G[\mathsf{cone}(t)]$ is replaced with a graph of the same $\ell$-folio. Therefore, the size of the part of $A^\star\cap B^\star$ that is disjoint from $\mathsf{bag}(x)$ should be at most the size of $(A\cap B)\setminus \mathsf{bag}(x)$. Therefore, $|A^\star\cap B^\star|\le k$.
     
    Hence, for the separation $(A^\star,B^\star)$ of $G_x$, we have that $|A^\star\cap B^\star|\le k$ and $|A^\star\cap \mathsf{bag}(x)|\ge |A\cap\mathsf{bag}(x)|>q$ and $|B^\star\cap \mathsf{bag}(x)|\ge |B\cap\mathsf{bag}(x)|>q$, which contradicts the fact that $\mathsf{bag}(x)$ is $(q,k)$-unbreakable in $G_x$.
\end{proof}

\subsection{Proof of \texorpdfstring{Theorem 1.1}{\autoref{main_theorem}}}
Now we can finally prove~\autoref{main_theorem} which we repeat for convenience.

\setcounter{section}{1}
\setcounter{theorem}{0}

\begin{theorem}
    There is an algorithm that, given a graph $G$ (with additional vertex colors) that excludes a graph $H$ as a topological minor, and an $\FODP$ formula~$\phi(\bar x)$ (over the colored graph vocabulary), decides whether $G\models \exists \bar x\,\varphi(\bar x)$ in time $f(H,\varphi)\cdot |V(G)|^{3}$, where $f$ is a computable function depending on $H$.
  Moreover, if $G\models \varphi(\bar v)$ for some $\bar v\in V(G)^{|\bar x|}$, the algorithm outputs such a tuple $\bar v$.
\end{theorem}

\setcounter{section}{6}

We set $\varphi^+:=\exists \bar{x}\,\varphi(\bar{x})$. We also assume that $\varphi^+$ is given in prenex normal form and we let $r$ be its quantifier rank. We use $n$ to denote $|V(G)|$ and $\Sigma$ to denote the signature of $G$.

    We set $r':=\max\{r,2|H|\}$.
    We also set 
    \begin{align*}
    k=f_k(r'), && \delta=b=q=d=q(k), && c_1:=f_{\mathsf{size}}(|\Sigma|,b,r'),\text{ and } && c_2:=p(|\Sigma|,b,r'),
    \end{align*}
    where $f_k$ and $f_{\mathsf{size}}$ are the functions of~\Cref{lemma_gadget}, $q$ is the function of~\Cref{thm:strong-unbreakability}, and~$p$ is the function of~\cref{lemma_rep_tw}.
    We finally set $c:=\max\{c_1,c_2\}$.

\subsubsection{The algorithm}
The algorithm proceeds in a bottom-up manner along the tree decomposition by computing representatives and doing replacements. Our ultimate goal is to compute a representative of $(G,V(G)^{r})$ of size at most $c$.

We first compute a strongly $(q,k)$-unbreakable decomposition $\mathcal{T}=(T,\mathsf{bag})$ of adhesion at most $d$ using the algorithm of~\cref{thm:strong-unbreakability}.
Then, we consider a reverse post-order traversal of $T$, i.e.,~an ordering of $V(T)$ such that each node is visited after all its (strict) descendants are visited.
We process the nodes by increasing order.
Before any iteration, we set $G^\bullet:=G$, $\bar{R}^\bullet=V(G)^r$, and $\mathcal{T}^\bullet:=\mathcal{T}$.

\paragraph{Iteration step}
We work with the graph $G^\bullet$ and its tree decomposition $\mathcal{T}^\bullet$.
Let $t$ be the node of $\mathcal{T}^\bullet$ that is currently processed.

Let $\mathsf{MaxAdh}(t)$ be the set of all children $z$ of $t$
such that there is no child $z'\neq z$ of~$t$ such that $\adh(z)\subseteq \adh(z')$.
Note that for every child $y$ of $t$ there is a $z\in\mathsf{MaxAdh}(t)$ such that $\adh(y)\subseteq \adh(z)$.
For every $z\in\mathsf{MaxAdh}(t)$, we set \[G_z^{\mathsf{Max}}:=\bigcup\{G^\bullet[\cone(y)]\mid \text{$y$ is a child of $t$ with $\adh(y)\subseteq \adh(z)$}\}\]
and $\mathbf{G}_z^\mathsf{Max}:=(G_z^{\mathsf{Max}},\adh(z),\rho)$, where $\rho$ is an arbitrary bijection from $\adh(z)$ to $|\adh(z)|$.

Using the algorithm of~\cref{lemma_rep_tw}, we compute a collection of pairs \[\mathcal{H}:=\{(z,(\mathbf{H}_z,\bar{R}_z))\mid z\in\mathsf{MaxAdh}(x)\},\] where for each $z\in\mathsf{MaxAdh}(x)$, $(\mathbf{H}_z,\bar{R}_z)$ is a $(\delta,r)$-representative of $(\mathbf{G}_z^{\mathsf{Max}},\bar{R}^\bullet\cap V(G_z^\mathsf{Max}))$ of size at most~$c_1$.

Let $(G',\bar{R}')$ be the tuple obtained from $(G^\bullet,\bar{R}^\bullet)$ by replacing, for each $z\in \mathsf{MaxAdh}(t)$,  $(\mathbf{G}_z^{\mathsf{Max}},\bar{R}^\bullet\cap V(G_z^{\mathsf{Max}}))$ with $(\mathbf{H}_z,\bar{R}_z)$.

Let $R_1',\ldots,R_r'\subseteq V(G')$ such that $\bar{R}'=(R_1',\ldots,R_r')$.
We set
\begin{itemize}
    \item $\mathbf{C}:= (G'\setminus \mathsf{comp}(t),\mathsf{adh}(t),\rho)$, where $\rho$ is an arbitrary bijection from $\mathsf{adh}(t)$ to~$|\mathsf{adh}(t)|$; 
    \item $\bar{Q}:=\left(R_1'\setminus\mathsf{comp}(t),\ldots,R_r'\setminus\mathsf{comp}(t)\right)$;
    \item $\mathbf{G}_t':=(G_t',\mathsf{adh}(t),\rho)$, where $G_t'$ is the subgraph of $G'$ induced by $\mathsf{bag}(t)$ and $\bigcup\{V(H_z) \mid z\in\mathsf{MaxAdh}(t)\}$; and
    \item $\bar{R}_t':=\left(R_1'\cap V(G_t'),\ldots,R_r'\cap V(G_t')\right)$.
\end{itemize}
We apply the algorithm of~\cref{lemma_gadget} for  $(\mathbf{G}_t',\bar{R}_t')$.
If this algorithm outputs a $(\delta,r)$-representative~$(\mathbf{G}^\star,\bar{R}^\star)$ of $(\mathbf{G}_t',\bar{R}')$, then we set $(G^\bullet,\bar{R}^\bullet):=(\mathbf{C},\bar{Q})\oplus(\mathbf{G}^\star,\bar{R}^\star)$.
We also update $\mathcal{T}^\bullet$ by removing all children of $t$ and setting $\mathsf{bag}(t)=V(G^\star)$. This finishes the iterative step.

\paragraph{Final step}
When all nodes have been processed, apply the algorithm of~\cref{thm_courcelle} to decide if~$G^\bullet$ satisfies~$\varphi^+$ when its quantified variables are interpreted in the corresponding sets~$R_i^\bullet$. Report the corresponding answer and if the answer is positive, then also return a tuple $\bar{v}$ of vertices, each picked from the $R_i^\bullet$ corresponding to the variables of $\bar{x}$. This is safe since $\bar{R}^\bullet\subseteq V(G)^r$.

\subsubsection{Proof of correctness}
We will show that the following invariants are maintained after each iteration. Assume that we are after the $i$th iteration. Then,
    \begin{enumerate}
        \item $G^\bullet$ excludes $H$ as a topological minor.

        \item $\mathcal{T}^\bullet$ is a tree decomposition with the following properties:
        \begin{enumerate}
            \item for every node $t$ of $\mathcal{T}^\bullet$ that is not yet processed, $G^\bullet[\mathsf{cone}(t)]$ admits a tight star decomposition of adhesion at most $d$;
            \item every node of $\mathcal{T}^\bullet$ that is not yet processed has the strong $(q,k)$-unbreakability property;
            \item for every node $t$ of $\mathcal{T}^\bullet$ that is already processed, we have $|\mathsf{cone}(t)|\le c$; and
        \end{enumerate}
        \item $\mathsf{type}_r(G,V(G)^r)=\mathsf{type}_r(G^\bullet,\bar{R}^\bullet).$
    \end{enumerate}
    We show the invariants by induction.
    We start by observing that in the beginning of each iteration, the underlying graph $G_z^\mathsf{Max}$ of $\mathbf{G}_z^\mathsf{Max}$ has treewidth at most $c+d$. This follows directly from the fact that every connected component of $G_z^\mathsf{Max}\setminus\mathsf{adh}(z)$ is a subgraph of $G[\mathsf{comp}(y)]$, for some child $y$ of $t$ with $\mathsf{adh}(y)\subseteq \mathsf{adh}(z)$ and that $|\mathsf{comp}(y)|\le c$; recall that because of the considered ordering, all children of $t$ have already been processed and therefore, by induction hypothesis, their cones have size at most $c$.
    Also, because of~\cref{lemma_compositionality} and the induction hypothesis,
    \[\mathsf{type}_r(G,V(G)^r)=\mathsf{type}_r(G',\bar{R}').\]
    Also, because of representatives being $h$-equivalent (because of~\Cref{obs_types_to_folios}) and~\cref{lem:excluded-top-minor}, $G'$ excludes $H$ as a topological minor.
    Because of representatives being $\delta$-equivalent and~\cref{obs_maintain_regularity},  every node of $\mathcal{T}^\bullet$ (that is not a child of $t$) has the strong $(q,k)$-unbreakability property, which shows invariant (2b). 
    We next show that for every node $t'$ of $\mathcal{T}^\bullet$ (that is not a child of $t$), $G_t$ admits a tight star decomposition of adhesion at most $d$. This will show invariant (2a) and that the requirements of~\cref{lemma_gadget} are satisfied.
    We fix a node $t'$ of $\mathcal{T}^\bullet$ (that is not a child of $t$).
    We consider the collection of connected components of $G^\bullet[\mathsf{comp}(t)]$.
    We say that two such components are equivalent if they are adjacent to exactly the same vertices in $\mathsf{bag}(t)$.
    We consider a partition of such components into equivalence classes with respect to this equivalence relation.
    This partition is used to define a star decomposition $\mathcal{S}=(S,\mathsf{bag}_{\mathcal{S}})$ as follows.
    Consider a star~$S$ with center $s$ and one leaf for each equivalence class; we denote by $\mathcal{C}(\ell)$ the equivalence class corresponding to the leaf $\ell$.
    We also define the mapping $\mathsf{bag}_\mathcal{S}$ that maps every vertex of $S$ to a subset of $G_t'$ in the following way:
    \begin{itemize}
        \item $\mathsf{bag}_\mathcal{S}(s):=\mathsf{bag}(t)$ and
        \item for every leaf $\ell$ of $S$,
        \[\mathsf{bag}_\mathcal{S}(\ell):=\bigcup_{C\in \mathcal{C}(\ell)}V(C)\cup N(C).\]
    \end{itemize}
    It is easy to see that $\mathcal{S}=(S,\mathsf{bag}_\mathcal{S})$ is a star decomposition of adhesion at most $d$, since for every leaf $\ell$ of $S$, $\mathsf{bag}_\mathcal{S}(\ell)\cap \mathsf{bag}_\mathcal{S}(s)$ is a subset of the adhesion (in $\mathcal{T}^\bullet$) of some node~$y$ of $\mathcal{T}^\bullet$. Also, by construction, $\mathcal{S}$ is tight.

    In order to show that the requirements of~\cref{lemma_gadget} are satisfied, we further observe that if $t$ is the currently processed node, the star decomposition obtained in the aforementioned way also has the following properties.
    \begin{itemize}
        \item For every leaf $\ell$ of $S$, $|\mathsf{bag}_\mathcal{S}(\ell)|\le c$. Indeed, since in $G'$, for every leaf $\ell$ of $S$, $\mathsf{bag}(\ell)\subseteq H_z$, for some $z\in\mathsf{MaxAdh}(t)$ and $|H_z|\le c$, we have $|\mathsf{bag}_\mathcal{S}(\ell)|\le c$.
        \item $\mathsf{bag}_\mathcal{S}(s)$ is $(q,k)$-unbreakable in $G_t'$; this follows directly from the fact that every node of $\mathcal{T}^\bullet$ (that is not a child of $t$) has the strong $(q,k)$-unbreakability property in~$G'$.
    \end{itemize}
    Also, since $G'$ excludes $H$ as a topological minor, the same holds for its subgraph $G_t'$. 

    Therefore, the requirements of~\cref{lemma_gadget} are satisfied and its application correctly gives a $(\delta,r)$-representative $(\mathbf{G}^\star,\bar{R}^\star)$ of $(\mathbf{G}_t',\bar{R}')$. Because of~\cref{lemma_compositionality}, $(\mathbf{C},\bar{Q})\oplus(\mathbf{G}^\star,\bar{R}^\star)$ and $(G',\bar{R}')$ have the same rank-$r$ annotated type, which shows invariant~(3). 
    Invariant~(1) follows from~\cref{lem:excluded-top-minor}.

\subsubsection{Running time}
The algorithm of~\cref{lemma_rep_tw} is applied in total a linear (in $n$) number of times. This algorithm is always invoked on graphs of treewidth at most $c+d$, and therefore the running time of each application is $\mathcal{O}_{|\Sigma|,b,r,c+d}(n) = \mathcal{O}_{|\Sigma|,b,r}(n)$. As for the algorithm of~\cref{lemma_gadget}, it runs in time $\mathcal{O}_{|H|,|\Sigma|,b,r,q,d,c}(|\mathsf{cone}(t)|^3)$. Since for every child $z$ of $t$, we have that $|\mathsf{cone}(z)|\le c$, we have that \[|\mathsf{cone}(t)|\le |\mathsf{mrg}(t)|+ c\cdot |\{z\in V(T)|\mid z\text{ is a child of $t$ in $T$}\}|.\]
Since the margins of the nodes in the initial tree decomposition $\mathcal{T}$ partition the vertex set of the input graph~$G$, we have that \[\sum_{t\in V(T)}|\mathsf{cone}(t)|\le n+c\cdot 2n = \mathcal{O}_{|\Sigma|,r}(n).\]
Therefore, the total running time of all calls to the algorithm of~\cref{lemma_gadget} is $\mathcal{O}_{|H|,|\Sigma|,r}(n^3)$.
This implies that the total running time of the iteration step is $\mathcal{O}_{|H|,|\Sigma|,r}(n^3)$. In the final step, the running time is $\mathcal{O}_{|H|,|\Sigma|,r}(n)$.

\section{Conclusion}
\label{sec:conclusion}

In this work, we have fully classified the subgraph-closed classes admitting efficient encoding of topological minors on which model checking for $\FODP$ is fixed-parameter tractable. A natural next question is to study the model checking problem also for dense graph classes that are not necessarily closed under taking subgraphs. 

Another interesting question is the following. In \cite{PilipczukSSTV22} the authors considered a framework where, after a polynomial time preprocessing, queries of separator logic can be answered in constant time. Is the same true for disjoint-paths logic? The most basic question is whether we can answer disjoint-paths queries in constant time (or even linear time) after preprocessing. 
Even though we failed to implement the framework for model checking, it may be the case that we can extend the framework for query answering after preprocessing, since we can incorporate the data structure computed by dynamic programming. 
Nevertheless, at this point, there are more difficulties because we do not know how to answer disjoint-paths queries in constant time in minor-closed classes. 
It seems plausible that we can lift results for bounded genus graphs to almost embeddable graphs and use the structure theorem to improve the running time for classes with excluded minors to linear. Then using our methods we would be able to improve it for all graphs to linear after preprocessing. Based on the nature of the irrelevant vertex technique, it seems unlikely that we can improve the query time to constant.

The last open question that we want to mention is whether, and how, $\FODP$ can be replaced by some more expressive fragment of \textsf{CMSO} in order to get a result similar to~\cref{main_theorem}.
As we already mentioned in the introduction, \cite{SauST25} introduced 
a fragment of \textsf{CMSO}, namely $\textsf{CMSO/tw+dp}$ where quantification is restricted to sets of bounded bidimensionality and it also permits the use of disjoint paths queries. It was proved in~\cite{SauST25} that model checking for $\textsf{CMSO/tw+dp}$ can be done in quadratic time in non-trivial  minor-closed graph classes.
Is there an adequate fragment of \textsf{CMSO}, similar to the one of \cite{SauST25}, that permits model checking on topological-minor-free graphs? An answer to this question is given in~\cite{Sau26}.

\vspace{-3mm}
\enlargethispage{\baselineskip}

\end{document}